\documentclass[journal]{IEEEtran} 
\usepackage[hidelinks]{hyperref}       % hyperlinks
\usepackage{hyperref}
\hypersetup{
	colorlinks=true,
	linkcolor=blue,
	citecolor=blue
}

\usepackage{cite}
\ifCLASSINFOpdf
  \usepackage[pdftex]{graphicx}
\else
  \usepackage[dvips]{graphicx}
\fi
\usepackage{amsmath}
\usepackage{algorithmic}
\usepackage{array}
\usepackage[caption=false,font=normalsize,labelfont=sf,textfont=sf]{subfig}
\let\MYorigsubfloat\subfloat
\renewcommand{\subfloat}[2][\relax]{\MYorigsubfloat[]{#2}}
\usepackage{url}
\usepackage{amsthm}
\usepackage{amssymb,cases}
\usepackage{arydshln}
\usepackage{algorithm}
\usepackage{nccmath}
\usepackage{multirow}
\usepackage{booktabs}

\DeclareMathOperator*{\argmin}{arg\,min}
\usepackage{color}

\newtheorem{theorem}{Theorem}%[section]
\newtheorem{lemma}{Lemma}
\newtheorem{proposition}{Proposition}
\newtheorem{assumption}{Assumption}

\def\d{\ldots}
\def\cK{\mathcal{K}}
\def\cM{\mathcal{M}}

\def\st{\mathrm{s.t.}}

\def\tr{\mathrm{tr}}
\def\rk{\mathrm{rank}}
\def\diag{\mathrm{Diag}}

\newcommand{\revbj}[1]{{\color{black}{#1}}}
\newcommand{\revxl}[1]{{\color{red}{#1}}}
\renewcommand{\v}[1]{\mathbf{#1}}
\newcommand{\vv}[1]{\boldsymbol{\mathbf{#1}}}
\newcommand{\m}[1]{\mathbf{#1}}

\def\hermitian{\dagger}
\def\transpose{\mathsf{T}}

\makeatletter
\newcommand{\pushright}[1]{\ifmeasuring@#1\else\omit\hfill$\displaystyle#1$\fi\ignorespaces}
\newcommand{\pushleft}[1]{\ifmeasuring@#1\else\omit$\displaystyle#1$\hfill\fi\ignorespaces}
\makeatother

\def\fk{~\forall\,k \in \cK}
\def\fm{~\forall\,m \in \cM}

\begin{document}
\title{\fontsize{23pt}{26pt}\selectfont An Adaptive Proximal Inexact Gradient Framework and Its Application to Per-Antenna Constrained Joint Beamforming and Compression Design}
 
\author{\IEEEauthorblockN{Xilai~Fan, Bo~Jiang, and Ya-Feng~Liu}% <-this % stops a space
\thanks{Part of this work has been presented at the IEEE International Conference on Acoustics, Speech, and Signal Processing (ICASSP) 2024 \cite{fan2024JointBeamformingCompression}. 
X. Fan and Y.-F. Liu are with the State Key Laboratory of Scientific and Engineering Computing, Institute of Computational Mathematics and Scientific/Engineering Computing, Academy of Mathematics and Systems Science, Chinese Academy of Sciences, Beijing 100190, China (e-mail: \{fanxilai, yafliu\}@lsec.cc.ac.cn). 
B. Jiang is with Ministry of Education Key Laboratory of NSLSCS, School of Mathematical Sciences, Nanjing Normal University, Nanjing 210023, China (e-mail: jiangbo@njnu.edu.cn). 
}}

\maketitle


\begin{abstract}
In this paper, we propose an adaptive proximal inexact gradient (APIG) framework for solving a class of nonsmooth composite optimization problems involving function and gradient errors. Unlike existing inexact proximal gradient methods, the proposed framework introduces a new line search condition that jointly adapts to function and gradient errors, enabling adaptive stepsize selection while maintaining theoretical guarantees. Specifically, we prove that the proposed framework achieves an $\epsilon$-stationary point within $\mathcal{O}(\epsilon^{-2})$ iterations for nonconvex objectives and an $\epsilon$-optimal solution within $\mathcal{O}(\epsilon^{-1})$ iterations for convex cases, matching the best-known complexity in this context. We then custom-apply the APIG framework to an important signal processing problem: the joint beamforming and compression problem (JBCP) with per-antenna power constraints (PAPCs) in cooperative cellular networks. This customized application requires careful exploitation of the problem’s special structure such as the tightness of the semidefinite relaxation (SDR) and the differentiability of the dual. Numerical experiments demonstrate the superior performance of our custom-application over state-of-the-art benchmarks for the JBCP.
\end{abstract}


\section{Background and Introduction}\label{apd:I}
\IEEEPARstart{I}n this paper, we consider the composite optimization problem of the form
\begin{equation}
	\min_{\v x\in\mathbb R^n} \  \left\{ F(\v x) := f(\v x) + h(\v x) \right\}, 
	\label{equ:generic_problem}
\end{equation}
where $h:  \mathbb{R}^n \to \mathbb{R}  \cup \{+\infty\}$ is a proper closed convex function with a tractable proximal mapping, and 
$f: \mathbb{R}^n \to \mathbb{R} $ is a continuously differentiable function that may be nonconvex. 
Problem \eqref{equ:generic_problem} has numerous practical applications in signal processing and machine learning (e.g., \cite{wright2009SparseReconstructionSeparable,hong2013JointBaseStation,zou2005RegularizationVariableSelection,candes2011RobustPrincipalComponent,liu2024SurveyRecentAdvances}) and includes several important classes as special cases. 
For instance, setting $h(\cdot) = 0$ in \eqref{equ:generic_problem} yields an unconstrained optimization problem, while setting $h(\cdot)$ as the indicator function of a (closed) convex set leads to a convex constrained optimization problem. 
In this paper, we are interested in the case where computing the function and gradient information of $f(\cdot)$ is costly, often requiring the solution of an associated convex subproblem. 
Examples of such problems can be found in \cite{liu2019NonergodicConvergenceRate, jiang2023RiemannianExponentialAugmented,wu2024CIBasedQoS,dai2018OptimizedBasestation,vanleeuwen2021VariableProjectionNonsmooth,poon2023SmoothOverparameterizedSolvers}.  

The proximal gradient (PG) method is an efficient class of algorithms for solving problem \eqref{equ:generic_problem}. 
Given a parameter $\lambda > 0$, the proximal operator $\operatorname{prox}_{\lambda h}: \mathbb{R}^n \rightarrow \operatorname{dom}(h) := \{\v x \in \mathbb{R}^n \mid h(\v x) < +\infty\}$ is defined as
\begin{equation}
	\operatorname{prox}_{\lambda h}(\v x) = \argmin_{\v y \in \mathbb R^n}\  \left\{\frac{1}{2\lambda} \|\v y - \v x\|^2 + h(\v y)\right\}. 
	\label{equ:def_prox}
\end{equation}
At the $i$-th iteration, the standard PG method for solving problem \eqref{equ:generic_problem} updates the next iterate as 
\begin{equation}
	\v x^{i+1} = \operatorname{prox}_{\lambda_i h}(\v x^{i} - \lambda_i \nabla f(\v x^i)). 
	\label{equ:PG}
\end{equation}
Here, $\lambda_i > 0$ is the stepsize, which typically satisfies certain line search (LS) condition. 
However, obtaining $\v x^{i+1}$ requires the exact knowledge of $\nabla f(\v x^i)$, which may be computationally expensive or even inaccessible in our interested setting. 
Consequently, it becomes crucial to solve problem \eqref{equ:generic_problem} using only inexact gradient or function information of $f(\cdot)$
\cite{bertsekas2000GradientConvergenceGradient, polyak1987IntroductionOptimization, schmidt2011ConvergenceRatesInexact, hamadouche2024SharperBoundsProximal, so2017NonasymptoticConvergenceAnalysis, khanh2024NewInexactGradient, hallak2024StudyFirstorderMethods, luo1993ErrorBoundsConvergence, cartis2018GlobalConvergenceRate, paquette2020StochasticLine, berahas2021GlobalConvergenceRate}. 
Let $\v g^i$ denote the inexact gradient at $\v x^i$. 
The proximal inexact gradient (PIG) method updates the iterate as 
\begin{equation}
	\v x^{i+1} = \operatorname{prox}_{\lambda_i h}(\v x^{i} - \lambda_i \v g^i), 
	\label{equ:PIG}
\end{equation}
which replaces the exact gradient $\nabla f(\v x^i)$ in the PG method \eqref{equ:PG} with its approximation $\v g^i$. 
To guarantee the convergence of the PIG method, the stepsize $\lambda_i$ in \eqref{equ:PIG} must satisfy a LS condition specifically tailored to the inexact function or gradient information. 
In the following, we review the existing results on the PIG method \eqref{equ:PIG}, including the selection strategies of the stepsize $\lambda_i$ and the required conditions on the inexactness of the function and gradient approximations. 

\subsection{Prior Works}

Given $\lambda > 0$ and $\v d\in \mathbb{R}^n$, the gradient mapping \cite{beck2017FirstOrderMethodsOptimization} at a point $\v x \in \mathbb{R}^n$ is defined as 
\begin{equation}
	\v G_{\lambda}(\v x, \v d) = \lambda^{-1} (\v x - \operatorname{prox}_{\lambda h}(\v x - \lambda \v d)). 
	\label{equ:def_gradient_mapping}
\end{equation}
By \cite[Theorem 10.7]{beck2017FirstOrderMethodsOptimization}, it holds that $\v G_{\lambda}(\v x, \nabla f(\v x)) = 0$ if and only if $\v x$ is a stationary point of problem \eqref{equ:generic_problem}.  
Denote the gradient error by 
\begin{equation}
	\varepsilon_i^{\v g} = \v g^i - \nabla f(\v x^i). 
	\label{equ:def_eg}
\end{equation}
Then various existing conditions on $\varepsilon_i^{\v g}$ can be unified in the following form: 
\begin{equation}
	\| \varepsilon_i^{\v g} \|^2 \leq (\eta_i^{\v g})^2 + (a^2 + b^2 \lambda_i^2) \varepsilon_i^2, 
	\label{equ:general_gradient_inexact}
\end{equation}
where $\{\eta_i^{\v g}\}$ is a predefined nonnegative sequence, $a \geq 0$, $b \geq 0$, and $\varepsilon_i$ is either $\|\v G_{\lambda_i}(\v x^i, \nabla f(\v x^i))\|$ or $\|\v G_{\lambda_i}(\v x^i, \v g^i)\|$. 
Note that, $\v G_{\lambda_i}(\v x^i, \nabla f(\v x^i))$ and $\v G_{\lambda_i}(\v x^i, \v g^i)$ \revbj{simplify} to $\nabla f(\v x^i)$ and $\v g^i$, respectively, when $h(\cdot) = 0$. 
When  $a = b = 0$ in \eqref{equ:general_gradient_inexact}, the corresponding condition is referred to as  the {\it absolute inexactness condition}, as the sequence $\{\eta_i^{\v g}\}$ is predefined. 
Otherwise, when $a + b > 0$, it is called the {\it relative inexactness condition}, as the gradient error $\varepsilon_i^{\v g}$ is adaptively controlled based on the information from the iterates.  
In the special case where $\eta_i^{\v g} \equiv 0$, $b = 0$, $\varepsilon_i = \|\v G_{\lambda_i}(\v x^i, \nabla f(\v x^i))\|$, and $h(\cdot) = 0$, the corresponding relative inexactness condition becomes
\begin{equation} \label{equ:normcondition}
	\|\varepsilon_i^{\v g}\| \leq a \|\nabla f(\v x^i)\|, \quad a \in (0,1), 
\end{equation}
which is known as the \textit{norm condition}, first introduced in \cite{polyak1987IntroductionOptimization} and later employed in other works, e.g., \cite{carter1991GlobalConvergenceTrust
	,berahas2021GlobalConvergenceRate, khanh2024NewInexactGradient}. 
In the following, we discuss both the absolute and relative inexactness conditions for the PIG method \eqref{equ:PIG}, under different stepsize strategies: constant stepsize strategies depending on the gradient Lipschitz constant and adaptive stepsize strategies based on LS conditions. 

\begin{table*}[t]\centering
	\setlength{\extrarowheight}{2pt}
	\tabcolsep=0.11cm
	\caption{Summary of the existing works on the PIG method. 
	}
	\begin{tabular}{c|c|cllc|cc|cc|cc}
		\hline
		\multirow{2}{*}{Literature} & \multirow{2}{*}{Stepsize} & \multicolumn{4}{c|}{$\|\varepsilon_i^{\v g}\|^2 \leq (\eta_i^{\v g})^2 + (a^2 + b^2 \lambda_i^2) \varepsilon_i^2 $} & \multicolumn{2}{c|}{$|\varepsilon_i^f| \leq \eta_i^f + c \lambda_i \|\v G_{\lambda_i}(\v x^i, \v g^i)\|^2$} & \multicolumn{2}{c|}{Objective}& \multicolumn{2}{c}{Iteration Complexity}   \\ \cline{3-12}
		&  & $\eta_i^{\v g}$ & $a$ & $b$ & $\varepsilon_i$ & $\eta_i^f$ & $c$ & $h(\cdot)$ & $f(\cdot)$  & NC & C   \\\hline
		\cite{schmidt2011ConvergenceRatesInexact}  & $1/L$ & summable & 0 & 0 & $\| \v G_{\lambda_i}(\v x^i, \v g^i) \|$ & 0 & 0 & C & C & -- & $\mathcal{O}(\epsilon^{-1})$\\
		\cite{khanh2024NewInexactGradient} & $1/L$ & 0 & $[0, 1)$ & 0 & $\| \nabla f(\v x^i) \|$ & 0 & 0 & 0 & NC & $^*$ & --\\ 
		\cite{hallak2024StudyFirstorderMethods} & $1/(L(1+a))$ & 0 & $[0, 1)$ & 0 & $\|\nabla f(\v x^i)\|$ & 0 & 0 & I & C & $\mathcal O(\epsilon^{-2})^\dagger$ & -- \\
		\cite{vanleeuwen2021VariableProjectionNonsmooth} & $1/L$ & 0 & 0 & $[0,L/2)$ & $\| \v G_{\lambda_i}(\v x^i, \v g^i) \|$ & 0 & 0 & C & NC & $\mathcal O(\epsilon^{-2})$ & -- \\
		\hline
		\cite{luo1993ErrorBoundsConvergence} & adaptive & 0 & $[0, 1)$ & 0 & $\| \v G_{\lambda_i}(\v x^i, \v g^i) \|$ & 0 & 0 & I & NC & linear$^\ddagger$ & -- \\
		\cite{cartis2018GlobalConvergenceRate}& adaptive & 0 & 0 & $[0, \infty)$ & $\|\v g^i\|$ & 0 & 0 & 0 & NC/C & 
		$\mathcal{O}(\epsilon^{-2})$ & $\mathcal{O}(\epsilon^{-1})$\\
		\cite{paquette2020StochasticLine} & adaptive & 0 & 0 & $[0, \infty)$ & $\|\v g^i\|$ & 0 & $\big[0, \frac{\theta \lambda_i}{4 \lambda_{\max}}\big]$ & 0 & NC/C & $\mathcal{O}(\epsilon^{-2})$ & $\mathcal{O}(\epsilon^{-1})$ \\
		\cite{berahas2021GlobalConvergenceRate} & adaptive & 0 & $[0, 1)$ & 0 & $\|\nabla f(\v x^i)\|$ & $\mathcal{O}(\epsilon^2)$ & 0 & 0 & NC/C & 
$\mathcal{O}(\epsilon^{-2})$ & $\mathcal{O}(\epsilon^{-1})$\\  \hline 
		\multirow{3}{*}{This work} & \multirow{3}{*}{adaptive} & square-summable & $[0,1)$ & $[0, \infty)$ &  $\| \v G_{\lambda_i}(\v x^i, \v g^i) \|$ &summable & $[0, \frac{\theta}{4}]$ & C & NC & $\mathcal{O}(\epsilon^{-2})$ & -- \\ 
		& & summable & 0 & 0 &   $\| \v G_{\lambda_i}(\v x^i, \v g^i) \|$ & summable & $[0, \frac{\theta}{4}]$  & C & C &
		-- & $\mathcal{O}(\epsilon^{-1})$ \\ 
		& & 0 &  $[0,1)$ & $[0, \infty)$  &   $\| \v G_{\lambda_i}(\v x^i, \v g^i) \|$ & 0  & $[0, \frac{\theta}{4}]$  & C & C &
		-- & $\mathcal{O}(\epsilon^{-1})$ \\
		\hline
	\end{tabular}
	\begin{minipage}{\textwidth}
		\vspace{6pt}
		\footnotesize
		\textit{Note:} In this table, I = indicator function of a convex set, NC = nonconvex, and C = convex. 
		The iteration complexity refers to the algorithm's complexity in returning an $\epsilon$-stationary point $\v x$ that satisfies $\|\v G_{\lambda}(\v x,\nabla f(\v x))\| \leq \epsilon$ in the NC case, and an $\epsilon$-optimal solution $\v x$ that satisfies $F(\v x) - F^\star \leq \epsilon$ in the C case. 
		The symbol ``--'' means that no corresponding results are established for that case; ``$*$'' indicates that \cite{khanh2024NewInexactGradient}  establishes  linear or sublinear convergence rates under an additional  K\L~condition  on $f(\cdot)$, ``$\dagger$'' indicates that the algorithm outputs an $(\epsilon+\mathcal{O}(\sqrt{a}))$-stationary point instead of an $\epsilon$-stationary point, and ``$\ddagger$'' indicates that an additional error bound assumption on $f(\cdot)$ is required. 
	\end{minipage}	\vspace{-12pt}
	\label{tab:differences}
\end{table*}

\textbf{Constant stepsize strategies.} For the case where $\nabla f(\cdot)$ is $L$-Lipschitz continuous on $\mathbb{R}^n$, the stepsize $\lambda_i$ can be chosen as a constant. 
Several works have analyzed the PIG method \eqref{equ:PIG} under the \textit{absolute inexactness condition}, namely the inexactness condition \eqref{equ:general_gradient_inexact} with $a = b = 0$. 
For $h(\cdot) = 0$, the work \cite{polyak1987IntroductionOptimization} analyzed the PIG method \eqref{equ:PIG} with a constant stepsize $\lambda_i \equiv \lambda \in (0, 2\mu L^{-3})$ for $\mu$-strongly convex $f(\cdot)$ with $\mu > 0$. 
It showed that under the absolute inexactness condition, when $\eta_i^{\v g} \equiv \eta$ for some constant $\eta > 0$, the sequence generated by the method converges linearly to a neighborhood of the optimal solution. 
The radius of this neighborhood depends on the constant $\eta$. 
Let $F^\star$ be the optimal value of problem \eqref{equ:generic_problem}. 
For convex $f(\cdot)$ and $h(\cdot)$, the work \cite{schmidt2011ConvergenceRatesInexact} showed that the PIG method \eqref{equ:PIG} with a constant stepsize \(\lambda_i \equiv 1/L\) returns an \(\epsilon\)-optimal solution, a point $\v x$ that satisfies \(F(\v x) - F^\star \leq \epsilon\), within \(\mathcal{O}(\epsilon^{-1})\) iterations under the absolute inexactness condition with summable \(\{\eta_i^{\v g}\}\), i.e., \(\sum_i \eta_i^{\v g} < +\infty\). 
This matches the complexity of the standard PG method in the absence of errors. 
Recently, the work \cite{hamadouche2024SharperBoundsProximal} improved the results in \cite{schmidt2011ConvergenceRatesInexact}, enabling a more flexible selection of \(\{\eta_i^{\v g}\}\) while ensuring the algorithm's convergence. 
Additionally, the complexity results in \cite{schmidt2011ConvergenceRatesInexact} were further analyzed in \cite{so2017NonasymptoticConvergenceAnalysis} under an error bound assumption for \(f\) in the case where \(h(\v x) = 0\). 

Under the \textit{relative inexactness condition}, namely, the inexactness condition \eqref{equ:general_gradient_inexact} with \(a + b > 0\), several studies have been conducted, including \cite{polyak1987IntroductionOptimization} and \cite{khanh2024NewInexactGradient} for the special case where \(h(\v x) = 0\) and \(\eta_i^{\v g} \equiv 0\), \cite{bertsekas2000GradientConvergenceGradient} for \(h(\v x) = 0\), and \cite{vanleeuwen2021VariableProjectionNonsmooth} for general convex $h(\cdot)$.  
Specifically, \cite{polyak1987IntroductionOptimization} established the linear convergence of the sequence generated by the PIG method \eqref{equ:PIG}
for $\mu$-strongly convex $f(\cdot)$ with $\mu>0$ and $h(\cdot) = 0$ under the norm condition \eqref{equ:normcondition} using a constant stepsize $\lambda_i \equiv \lambda \in \left(0, 2\mu(1-a)/\left(L^2(1+a^2)\right)\right)$.  
The work \cite{khanh2024NewInexactGradient} proposed a practical PIG method \eqref{equ:PIG} with \(\lambda_i \equiv 1/L\), which introduces a strategy to ensure the norm condition \eqref{equ:normcondition}, even without access to the exact gradient \(\nabla f(\v x^i)\). 
Additionally, the work \cite{khanh2024NewInexactGradient} established the convergence rate of $\|\nabla f(\v x^i)\|$ under the assumption that $f(\cdot)$ satisfies the Kurdyka-\L ojasiewicz (K\L) condition. 
The work \cite{hallak2024StudyFirstorderMethods} considered the norm condition \eqref{equ:normcondition} in the PIG method \eqref{equ:PIG} with a constant stepsize $\lambda_i \equiv 1/(L(1+a))$, where $h(\cdot)$ is an indicator function of a convex set and $f(\cdot)$ is nonconvex. 
It showed that the corresponding PIG method produces a point \(\v x\) satisfying \(\| \v G_{\lambda}(\v x, \nabla f(\v x))\| \leq \epsilon +  \mathcal{O}(\sqrt{a})\) within \(\mathcal{O}(\epsilon^{-2})\) iterations. 
For nonconvex $f(\cdot)$ and convex $h(\cdot)$, \cite{vanleeuwen2021VariableProjectionNonsmooth} showed that the PIG method \eqref{equ:PIG} with \(\lambda_i \equiv 1/L\), under the inexactness condition \eqref{equ:general_gradient_inexact} with \(a = 0\), \(b \in (0, L/2)\), and \(\varepsilon_i = \|\v G_{\lambda_i}(\v x^i, \v g^i)\|\), achieves an \(\epsilon\)-stationary point, that satisfies \(\|\v G_{\lambda_i}(\v x^i, \v g^i)\| \leq \epsilon\), within \(\mathcal{O}(\epsilon^{-2})\) iterations. 

However, the constant stepsizes used in all the above methods (e.g.,  \cite{vanleeuwen2021VariableProjectionNonsmooth, polyak1987IntroductionOptimization, schmidt2011ConvergenceRatesInexact, so2017NonasymptoticConvergenceAnalysis, hamadouche2024SharperBoundsProximal, khanh2024NewInexactGradient, hallak2024StudyFirstorderMethods}) depend on the Lipschitz constant $L$, which is typically unknown. 
This issue is particularly critical in our case, where computing $f(\v x)$ or $\nabla f(\v x)$ is already challenging. 
Consequently, the PIG method \eqref{equ:PIG} with constant stepsizes becomes ineffective. 
To address this limitation, a LS condition that relies on the inexact gradient or function information is required. 

\textbf{Adaptive stepsize strategies.} The existing LS conditions can be divided into two categories, depending on whether exact function information is available or not: (i) conditions that use exact function information \cite{luo1993ErrorBoundsConvergence,cartis2018GlobalConvergenceRate}, and (ii) conditions that use approximate function information \cite{berahas2021GlobalConvergenceRate, paquette2020StochasticLine}. 
Among those in the first category, the work \cite{luo1993ErrorBoundsConvergence} proposed a general framework that adaptively selects the stepsize, ensuring a sufficient decrease in the exact function value at the next iterate. 
It considered the case where $h(\cdot)$ is the indicator function of a convex set and $f(\cdot)$ is nonconvex. 
The framework addresses the relative inexactness condition \eqref{equ:general_gradient_inexact} with $\eta_i^{\v g} \equiv 0, a \in (0,1), b = 0,$ and $\varepsilon_i = \|\v G_{\lambda_i}(\v x^i, \v g^i)\|$. 
Additionally, it established the linear convergence of the corresponding algorithm under certain error bound conditions.
The work \cite{cartis2018GlobalConvergenceRate} replaced the exact gradient in the classical Armijo LS condition with its approximations and established the convergence and convergence rate of the corresponding PIG method \eqref{equ:PIG}. 
The analysis in \cite{cartis2018GlobalConvergenceRate} focused on the case $h(\cdot) = 0$, with $f(\cdot)$ being either nonconvex or convex, under the inexactness condition \eqref{equ:general_gradient_inexact} with $\eta_i^{\v g} \equiv 0, a = 0, b \in [0,\infty) $, and $ \varepsilon_i = \|\v g^i\|$. 

In the second category, 
the works \cite{paquette2020StochasticLine} and \cite{berahas2021GlobalConvergenceRate} extended the result in \cite{cartis2018GlobalConvergenceRate} to scenarios where the exact function value is unavailable. 
Both works focused on the case $h(\cdot) = 0$ with $f(\cdot)$ being convex or nonconvex. 
Let $f_i$ denote the inexact function value at $\v x^i$, and let the function error be 
\begin{equation}
	\varepsilon_i^f = f_i - f(\v x^i). 
	\label{equ:def_eif}
\end{equation}
The work \cite{paquette2020StochasticLine} used the same inexactness condition and the LS condition as \cite{cartis2018GlobalConvergenceRate}, but introduced the following inexactness condition for the function error \revbj{$\varepsilon_i^f$}: 
\begin{equation}
	|\varepsilon_i^f| \leq c \lambda_i \|\v g^i\|^2,  \label{equ:paquette}
\end{equation}
where $c \in [0, \theta \lambda_i/(4\lambda_{\max})]$ with $\theta$ being a parameter in the Armijo LS condition (see e.g., \eqref{equ:Armijo} and \eqref{equ:B1}) and $\lambda_{\max}$ being the upper bound of the stepsize. 
The work \cite{berahas2021GlobalConvergenceRate} used the norm condition \eqref{equ:normcondition} in the PIG method and analyzed the case where the function error $\varepsilon_i^f$ is bounded by a constant multiple of the desired accuracy $\epsilon$. 
They modified the LS condition proposed in \cite{cartis2018GlobalConvergenceRate} by incorporating an additional term equal to twice the function's error bound. 
\revbj{Theoretically}, the works \cite{cartis2018GlobalConvergenceRate, paquette2020StochasticLine, berahas2021GlobalConvergenceRate}, 
which specifically focus on the case $h(\cdot) = 0$, 
demonstrated that the PIG method achieves an $\epsilon$-stationary point for nonconvex $f(\cdot)$ within $\mathcal{O}(\epsilon^{-2})$ iterations and an $\epsilon$-optimal solution for convex $f(\cdot)$ within $\mathcal{O}(\epsilon^{-1})$ iterations, under the considered inexact conditions. 

We summarize the aforementioned works in Table \ref{tab:differences} for a comparative analysis. 
The table shows that, while existing PIG methods address gradient errors under absolute or relative inexactness conditions, they often assume either a zero function error or a function error bounded by a sufficiently small constant relative to the desired accuracy. 
A notable exception is \cite{paquette2020StochasticLine}, though it is restricted to the smooth setting ($h(\cdot) = 0$). 
Such strong assumptions make these methods impractical for our problem setting.

To overcome this limitation, we consider a broader scenario where both function and gradient errors are controlled under more general inexactness conditions. 
For the gradient error, we adopt \eqref{equ:general_gradient_inexact} with \(\varepsilon_i = \|\v G_{\lambda_i}(\v x^i, \v g^i)\|\), while for the function error, we propose the following \emph{new} inexactness condition:  
\begin{equation}
	|\varepsilon_i^f| \leq \eta_i^f + c \lambda_i \|\v G_{\lambda_i}(\v x^i, \v g^i)\|^2,  
	\label{equ:general_function}
\end{equation}  
where \(\{\eta_i^f\}\) is a predefined summable sequence and \(c \in [0, \theta/4]\).  
Unlike prior works including \cite{paquette2020StochasticLine}, which focused specifically on \(h(\v x) = 0\), our approach applies to a general convex $h(\cdot)$. 
Even in the special case \revbj{where} \(h(\v x) = 0\), our proposed inexactness condition \eqref{equ:general_function} is weaker than \eqref{equ:paquette} in \cite{paquette2020StochasticLine}, as it allows a broader range for the parameter \(c\) due to \(\lambda_i \leq \lambda_{\max}\) and enables more relaxed inexactness conditions through the choice of \(\eta_i^f\). 
For clarity, we include our main results in Table \ref{tab:differences} to highlight these distinctions.

\subsection{Our Contributions}
In this paper, we propose an adaptive PIG (APIG) framework, which allows for both function and gradient errors, to solve problem \eqref{equ:generic_problem}, and custom-apply it to the joint beamforming and compression problem (JBCP) with per-antenna power constraints (PAPCs) introduced in \cite{fan2024JointBeamformingCompression}. 
The JBCP with PAPCs is an important signal processing problem in the cooperative cellular networks, as it facilitates the joint design of the transmission strategies for the base stations (BSs) while utilizing the fronthaul links under practical PAPCs to improve the overall communication throughput\cite{fan2024JointBeamformingCompression, fan2022EfficientlyGloballySolving, fan2023QoSbasedBeamformingCompression}. 
\revbj{Our contributions are twofold}:  
\begin{itemize}
	\item \emph{New APIG Framework}. 
	We propose an APIG framework to solve problem \eqref{equ:generic_problem}, which allows for both function and gradient errors under absolute and relative inexactness conditions. 
	A key feature of the framework is a new LS condition that enables adaptive stepsize selection and effectively handles both function and gradient errors. 
	Compared to existing methods, our framework accommodates a broader class of $h(\cdot)$ and ensures convergence results under weaker inexactness conditions. 
	Specifically, we establish that the framework can return an \(\epsilon\)-stationary point within \(\mathcal{O}(\epsilon^{-2})\) iterations for nonconvex $f(\cdot)$ and an \(\epsilon\)-optimal solution within \(\mathcal{O}(\epsilon^{-1})\) iterations for convex $f(\cdot)$, which match the best-known complexity in cases where both function value and gradient involve errors. 
 	\item \emph{Customized Application to JBCP with PAPCs}.
	By establishing the tightness of the semidefinite relaxation (SDR) of the JBCP with PAPCs, we reformulate it as an equivalent \textit{differentiable} dual problem in the form of problem \eqref{equ:generic_problem} via the Lagrangian dual of the SDR. 
	The differentiability is crucial as it enables the use of gradient-based algorithms, which significantly outperform slower subgradient-based algorithms (e.g., \cite{yu2007TransmitterOptimizationMultiantenna, dartmann2013DualityMaxminBeamforming, zhang2020DeepLearningEnabled, miretti2024ULDLDualityCellfree}).
	We then custom-apply the proposed APIG framework to solve the dual problem of the JBCP with PAPCs, leading to APIG-FP,~where an efficient fixed point (FP) algorithm is employed to compute the function and gradient approximations of the dual problem in a controllable fashion, and establish the convergence of APIG-FP under mild conditions. 
	Numerical experiments validate the efficiency of the customized~application of the APIG algorithm against the state-of-the-art benchmarks. 
\end{itemize}

In our prior work \cite{fan2024JointBeamformingCompression}, we introduced PIGA, an efficient inexact projected gradient method for solving the JBCP with PAPCs. 
This paper significantly extends \cite{fan2024JointBeamformingCompression} in three aspects: 
(i) developing a new APIG framework with proven iteration complexities for more general problem \eqref{equ:generic_problem}, 
(ii) proving the convergence of APIG-FP (tailored for the JBCP with PAPCs)  under mild conditions, and
(iii) incorporating function and gradient errors into APIG-FP’s LS condition. 

\emph{Notations.} We adopt the following notations throughout this paper. 
For any vector $\v x$, $\diag(\v x)$ denotes the diagonal matrix with the elements of $\v x$ on its main diagonal, $\|\v x\|$ denotes the $\ell_2$-norm of $\v x$, and inequalities between vectors are understood componentwise. 
For any vector $\v p = [p_1, p_2, \d, p_K]^\transpose$ and $\v q = [q_1, q_2, \d, q_K]^\transpose$, the Thompson's metric \cite{lemmens2012NonlinearPerronFrobeniusTheory} $ \mu: \mathbb{R}_{++}^K \times \mathbb{R}_{++}^K \rightarrow \mathbb{R}_+^K $ is defined as 
\begin{equation}
	\mu(\v p, \v q) = \max_{k \in \{1,2,\d,K\}} \left| \log_{\mathrm{e}} ( p_k) - \log_{\mathrm{e}} ( q_k)\right|. 
	\label{equ:metric}
\end{equation}
For any matrix $\m{A}$, $\m{A}^\hermitian$ and $\m{A}^\transpose$ denote the conjugate transpose and transpose of $\m{A}$, respectively;
%$\rk(\m{A})$ denotes the rank of $\m{A}$; 
$\m{A}^{(m, n)}$ denotes the entry on the $m$-th row and the $n$-th column of $\m{A}$. 
We use $\m{0}$ to denote the all-zero matrix of an appropriate size and $\m E_m$ to denote the all-zero square matrix except its $m$-th diagonal entry being one. 

\section{APIG Framework}

In this section, we develop an APIG framework for solving problem \eqref{equ:generic_problem}. 
Specifically, we propose new LS conditions which incorporate both function and gradient errors in Section \ref{apd:ls_ic}. Then we present the APIG framework and establish its convergence in Section \ref{apd:pafca}.

We make the following assumption on the function $f(\cdot)$ in problem \eqref{equ:generic_problem} throughout this section. 
\begin{assumption}\label{ass:1}
	The gradient $\nabla f(\cdot)$ is $L$-Lipschitz with $L > 0$ on $\operatorname{dom}(h)$, i.e., 
	$$
	\| \nabla f(\v x) - \nabla f(\v y) \| \leq L \|\v x - \v y\|,\quad\forall\,\v x, \v y \in \operatorname{dom}(h). 
	$$
\end{assumption}

To ensure the convergence of the proposed APIG framework, we next investigate the appropriate LS and inexactness conditions to effectively control the errors in both function and gradient evaluations. 

\subsection{LS and Inexactness Conditions}\label{apd:ls_ic}

As discussed in Section \ref{apd:I}, the inexactness conditions, which control errors in both function and gradient evaluations, along with appropriate LS conditions, play central roles in ensuring the convergence of the corresponding algorithms. In this subsection, we investigate the relevant LS and inexactness conditions. We begin by considering the case where exact function and gradient information are available. 

For the PG method \eqref{equ:PG} which has access to exact function and gradient information, the stepsize $\lambda_i$ typically satisfies  one of the following sufficient decrease conditions. 
The first condition, usually used for nonconvex $f(\cdot)$  \cite{wright2009SparseReconstructionSeparable}, ensures a sufficient decrease in the overall objective function $F(\cdot)$: 
\begin{equation}
	\begin{aligned}
		&F(\v x^{i+1}) \leq F(\v x^i) - \frac{\theta}{\lambda_i}\|\v x^{i+1} - \v x^i\|^2,
	\end{aligned}
	\label{equ:Armijo}
\end{equation}
where $\theta \in (0, 1)$. 
The second one, commonly applied to convex $f(\cdot)$ \cite{beck2017FirstOrderMethodsOptimization}, 
relies only on the function $f(\cdot)$ rather than the entire objective $F(\cdot)$: 
\begin{equation}
	f(\v x^{i+1}) \leq f(\v x^i) + \langle \nabla f(\v x^i), \v x^{i+1} - \v x^i\rangle + \frac{1}{2\lambda_i} \|\v x^{i+1} - \v x^i\|^2. 
	\label{equ:B2_or}
\end{equation}
It can be shown that the second condition implies the first one with $\theta = 1/2$. 
See the discussion in Appendix \ref{apd:finite_termination} in the Supplementary Material for details.  

When gradient or function errors are present, the LS conditions \eqref{equ:Armijo} and \eqref{equ:B2_or}
that depend on exact function and gradient evaluations become inapplicable. 
To address this limitation, modified LS conditions are required to effectively handle errors in both function and gradient evaluations. 
Next, we provide a detailed explanation of the modifications made on \eqref{equ:Armijo} and \eqref{equ:B2_or}. 

Let $\lambda$ be a trial stepsize and a candidate for $\lambda_i$. Define 
\begin{equation}
	\v{x}^i(\lambda) = \operatorname{prox}_{\lambda h}(\v{x}^{i} - \lambda \v{g}^i),
	\label{equ:x:lambda}
\end{equation}
and let $f_i(\lambda)$ represent the inexact function value at $\v{x}^i(\lambda)$. 
The corresponding function error at $\v{x}^i(\lambda)$ is defined as 
\begin{equation}
	\varepsilon_i^f(\lambda) = f_i(\lambda) - f(\v{x}^i(\lambda)). 
	\label{equ:def_eif_lambda}
\end{equation}  
For the gradient error \(\varepsilon_i^{\v{g}}\) defined in \eqref{equ:def_eg}, we impose
\begin{equation}
	\|\varepsilon_i^{\v{g}}\|^2 \leq (\eta_i^{\v{g}})^2 + \left( a^2 \lambda^{-2} + b^2 \right) \|\v{x}^i(\lambda) - \v{x}^i\|^2,
	\label{equ:general_gradient_inexact_new}
\end{equation}
where \(\{\eta_i^{\v{g}}\}\) is a nonnegative, square-summable sequence satisfying
\begin{equation}
	\sum_{i=0}^\infty (\eta_i^{\v{g}})^2 < \infty,
	\label{equ:def_etai}
\end{equation}
and \(a \in [0, 1)\), \(b \geq 0\). 
For the function errors  \(\varepsilon_i^f\) in \eqref{equ:def_eif} and  \(\varepsilon_i^f(\lambda)\) in \eqref{equ:def_eif_lambda}, 
we consider
\begin{equation}
	|\varepsilon_i^f| \leq \eta_i^f + c \lambda^{-1} \|\v{x}^i(\lambda) - \v{x}^i\|^2
	\label{equ:function_error1}
\end{equation}
and 
\begin{equation}
	|\varepsilon_i^f(\lambda)| \leq \eta_{i}^f +  c \lambda^{-1} \|\v{x}^i(\lambda) - \v{x}^i\|^2,
	\label{equ:function_error2}
\end{equation}
where  \(c \in [0, \theta/4]\) and \(\{\eta_i^f\}\) is a nonnegative, summable sequence satisfying
\begin{equation}
	\sum_{i=0}^\infty \eta_i^f < \infty.
	\label{equ:def_etaif}
\end{equation} 

We propose the following LS conditions to address the presence of both function and gradient errors. 
Specifically, the stepsize $\lambda_i$ is required to satisfy one of the following conditions.  
The first one is 
\begin{equation}
	f_i(\lambda) + h(\v x^i(\lambda)) \leq f_i + h(\v x^i) - \frac{\theta}{\lambda}\| \v x^{i}(\lambda) - \v x^i \|^2 + \nu_i, 
	\label{equ:B1}
\end{equation}
while the second one is 
\begin{equation}
	f_i(\lambda) \leq f_i + \left\langle\v g^i, \v x^{i}(\lambda) - \v x^i \right\rangle + \frac{1}{2\lambda} \left\|\v x^{i}(\lambda) - \v x^i\right\|^2 + \nu_i,
	\label{equ:B2}
\end{equation}
where
\begin{equation}
	\nu_i = \Upsilon_1(\lambda) (\eta_i^{\v g})^2 + 2 \eta_i^f
	\label{equ:def_nu}
\end{equation}
and  
\begin{equation}
	\Upsilon_1(\lambda) = 
	\begin{cases}
		\frac{1}{2}, & \text{~if~} a^2 + b^2= 0, \\
		\min\left\{\frac{\lambda}{2a}, \frac{1}{2b}\right\}, & \text{~otherwise}.
	\end{cases}
	\label{equ:cond_C}
\end{equation}

Compared to \eqref{equ:Armijo} and \eqref{equ:B2_or}, LS conditions \eqref{equ:B1} and \eqref{equ:B2} handle function and gradient errors by the relaxation term $\nu_i$ in \eqref{equ:def_nu}, which is specially designed based on the errors in the function and gradient evaluations.
If the trial stepsize \(\lambda\) is accepted, then we set \(\lambda_i = \lambda\) and $\v x^{i+1} = \v x^i(\lambda_i)$, which, together with \eqref{equ:x:lambda}, gives 
\[
\v{x}^{i+1} = \operatorname{prox}_{\lambda_i h}(\v{x}^i - \lambda_i \v{g}^i).
\]
Combining this with \eqref{equ:def_gradient_mapping} further implies $\v x^i - \v x^{i+1} = \lambda_i {\v G}_{\lambda_i}(\v x^i, \v g^i)$. 
Consequently, the gradient and function error conditions in  \eqref{equ:general_gradient_inexact_new} and \eqref{equ:function_error1}
	are consistent with those in \eqref{equ:general_gradient_inexact} and \eqref{equ:general_function} (with $\varepsilon_i = \v G_{\lambda_i}(\v x^i, \v g^i)$), respectively.
Moreover, when $h(\cdot) = 0$, $\eta_i^{\v g}  \equiv \eta_i^f \equiv 0$, and $a = 0$,   our gradient error condition  \eqref{equ:general_gradient_inexact_new}  simplifies to 
\begin{equation} \label{equ:grad:error:h=0}
	\|\varepsilon_i^{\v g}\| \leq   b \lambda_i \| \v g^i\|,  \quad b \geq 0, 
\end{equation}
which is used in \cite{luo1993ErrorBoundsConvergence} and \cite{cartis2018GlobalConvergenceRate}, 
while the function error condition  \eqref{equ:function_error1} coincides with \eqref{equ:paquette} in \cite{paquette2020StochasticLine}. However, these existing results are restricted to smooth problems ($h(\cdot) = 0$).
Our conditions generalize these results to composite nonsmooth objectives ($h(\v x) \neq 0$) while allowing for more flexible inexactness through the introduction of $\eta_i^f$, $\eta_i^{\v g}$, and $a$.

Define the upper bound of the stepsize as 
%\begin{equation}
%	\overline{\lambda}_{\theta} = 
%	\begin{cases}
%		\frac{2(1-\theta-2c)}{L+1}, & \text{if  $a^2 + b^2$ = 0}, \\[5pt]
%		\frac{2(1-\theta-a-2c)}{L+2b+1}, &  \text{otherwise}.
%	\end{cases}\label{equ:def_tildelambda:theta}
%\end{equation}
\begin{equation}
	\overline{\lambda}_{\theta} = 
		\frac{2(1-\theta-a-2c)}{L+2b+1}. \label{equ:def_tildelambda:theta}
\end{equation}
We now show that the LS conditions \eqref{equ:B1} and \eqref{equ:B2} hold for some suitably chosen $\lambda$ under the following mild assumption.
\begin{assumption}\label{ass:2}
	The function f in (1) is level-bounded, i.e., for any $t \in \mathbb{R}$, the set $\{\v x \in \mathbb{R}^n \mid f(\v x) \leq t\}$ is bounded. 
	In addition, the optimal set of problem (1) is nonempty and its corresponding optimal set and value are denoted by $\mathcal X^\star$ and $F^\star$, respectively.
\end{assumption}

\begin{lemma}\label{lem:ls}
Suppose that Assumptions \ref{ass:1} and \ref{ass:2} hold. 
	Let the parameters satisfy $0 < a < 1 - 2c - \theta$ for the LS condition \eqref{equ:B1} and $0 < a < 1- 4c$ for the LS condition \eqref{equ:B2}. 
	Then, the following statements hold: 
	\begin{itemize}
		\item[(i)] The LS condition \eqref{equ:B1} holds for all $\lambda \in (0, \overline{\lambda}_{\theta}]$. 
		\item[(ii)] The LS condition \eqref{equ:B2} holds for all $\lambda \in (0, \overline{\lambda}_{1/2}]$. 
	\end{itemize}
\end{lemma}
\begin{proof}
	See Appendix~\ref{apd:ls} in the Supplementary Material.  
\end{proof}

Based on Lemma \ref{lem:ls}, we adopt a simple backtracking strategy to select a suitable stepsize that satisfies condition \eqref{equ:B1} or \eqref{equ:B2}. 
Given  \(\lambda_{\max} > \lambda_{\min} > 0\), the initial trial stepsize $\lambda_{i}^{(0)}$ is chosen as: for $i = 0, 1$, any value in $ [\lambda_{\min}, \lambda_{\max}]$, and for $i \geq 2$
\begin{equation}
    \label{equ:BB}
    \lambda_{i}^{(0)} = 
    	\min\left\{\max\left\{\lambda^{\operatorname{ABB}}_{i-1}, \lambda_{\min}\right\}, \lambda_{\max}\right\}, 
\end{equation}
where $\lambda_i^{\mathrm{ABB}}$ is the alternate Barzilai-Borwein (ABB) stepsize \cite{dai2005ProjectedBarzilaiBorweinMethods} defined as 
\begin{equation}
    \label{equ:alterBB}
    \lambda^{\operatorname{ABB}}_i = 
    \begin{cases}
        \frac{\|\v x^{i} - \v x^{i-1}\|^2}{\left|(\v x^{i} - \v x^{i-1})^\transpose (\v g^{i} - \v g^{i-1})\right|}, & \text{if } i \text{ is even}, \\[5pt]
        \frac{\left|(\v x^{i} - \v x^{i-1})^\transpose (\v g^{i} - \v g^{i-1})\right|}{\|\v g^{i} - \v g^{i-1}\|^2}, & \text{otherwise}.
    \end{cases}
\end{equation}
Note that, unlike the original ABB stepsize in \cite{dai2005ProjectedBarzilaiBorweinMethods},  ours in \eqref{equ:alterBB} is computed using the inexact gradient information rather than the exact one.  
We then iteratively reduce the stepsize by selecting the smallest nonnegative integer  \(\ell\) such that \(\lambda_i = \lambda_{i}^{(0)} \alpha^\ell\) satisfies the LS condition \eqref{equ:B1} or \eqref{equ:B2}, where \(\alpha \in (0,1)\) is a predefined decreasing ratio parameter. 
The following proposition establishes the validity of this strategy, and its proof is omitted since regular. 

\begin{proposition}\label{lem:bound}
Under the same assumptions and parameter settings as in Lemma \ref{lem:ls}, 
let $\underline{\lambda}_{\theta} = \min\{\lambda_{\min}, \alpha \overline{\lambda}_{\theta}\}$. 
%, where \(\overline{\lambda}_{\theta}\) is given in Lemma \ref{lem:ls}. 
Then, there always exists an integer \(\ell_\theta \) satisfying \(0 \leq \ell_\theta \leq \lceil \log_{\alpha} (\lambda_{\max} / \underline{\lambda}_{\theta}) \rceil\) such that the stepsize \(\lambda_i = \lambda_{i}^{(0)} \alpha^{\ell_\theta }\) satisfies the LS conditions \eqref{equ:B1} or \eqref{equ:B2}. 
Furthermore, we have
\begin{equation}
	\lambda_i \in [\underline{\lambda}_{\theta}, \lambda_{\max}]. 
	\label{equ:lambda_bd}
\end{equation}
\end{proposition}

To conclude this subsection, we highlight some key distinctions of our proposed LS conditions compared to existing works. 
First, existing LS conditions based on inexact function or gradient evaluations for the PIG method \eqref{equ:PIG} mainly focus on the unconstrained case where \(h(\v x) = 0\), such as those in \cite{cartis2018GlobalConvergenceRate,paquette2020StochasticLine,berahas2021GlobalConvergenceRate}. 
In contrast, we consider a more general setting with a convex $h(\cdot)$, significantly extending the applicability of the PIG method to constrained and nonsmooth optimization problems. Second, even in the special case where \(h(\v x) = 0\), our LS conditions are more flexible  than those in \cite{cartis2018GlobalConvergenceRate,paquette2020StochasticLine,berahas2021GlobalConvergenceRate}. 
In this case, the LS condition \eqref{equ:B1} and the function error conditions \eqref{equ:function_error1} and \eqref{equ:function_error2} simplify to
\begin{equation}
    f_i(\lambda) \leq f_i - \theta \lambda \|\v g^i\|^2 + \nu_i
    \label{equ:B1:h=0}
\end{equation}
and 
\begin{equation}
	\label{equ:f:error:unc}
	|\varepsilon_i^f| \leq \eta_i^f + c \lambda  \|\v g^i\|^2, \quad 	|\varepsilon_i^f(\lambda)| \leq \eta_i^f + c \lambda  \|\v g^i\|^2.
\end{equation}
Notably, the LS condition in \cite{paquette2020StochasticLine} is a special case of \eqref{equ:B1:h=0} and \eqref{equ:f:error:unc} with \(\nu_i \equiv 0\) and \(\eta_i^f \equiv 0\); 
the LS condition in \cite{cartis2018GlobalConvergenceRate} does not account for function errors, 
which corresponds to setting \(\eta_i^f \equiv 0\) and \(c = 0\) in \eqref{equ:f:error:unc}; 
and the LS condition in \cite{berahas2021GlobalConvergenceRate} corresponds to  \eqref{equ:B1:h=0} with \(\nu_i = \mathcal{O}(\epsilon^2)\), 
where \(\epsilon\) represents the algorithm's tolerance.

\subsection{Proposed Algorithmic Framework}\label{apd:pafca}
Before presenting the APIG framework, we first discuss its termination condition. 
The goal is to find an approximate $\epsilon$-stationary point of problem \eqref{equ:generic_problem}, defined by $\|\v G_{\lambda}(\v x, \nabla f(\v x))\| \leq \epsilon$. 
Since exact gradient information is unavailable in our case, it is necessary to characterize the relationship between $\|\v G_{\lambda_i}(\v x^i, \nabla f(\v x^i))\|$ and its approximation $\|\v G_{\lambda_i}(\v x^i, \v g^i)\|$.

\begin{lemma}\label{lem:KKT}
	For each iterate $\v x^i$, there holds 
	$$
	\|\v G_{\lambda_i}(\v x^i, \nabla f(\v x^i)) \|\leq \| \v G_{\lambda_i}(\v x^i, \v g^i) \| + \|\varepsilon_i^{\v g}\|. 
	$$
\end{lemma}
\begin{proof}
	From the definition of the gradient mapping in \eqref{equ:def_gradient_mapping} and the firm nonexpansiveness of the proximal operator \cite[Theorem 6.42]{beck2017FirstOrderMethodsOptimization}, we have
	\begin{equation*}
		\begin{aligned}
			&\quad \| \v G_{\lambda_i}(\v x^i, \nabla f(\v x^i)) - \v G_{\lambda_i}(\v x^i, \v g^i) \| \\
			&= \lambda_i^{-1} \| \operatorname{prox}_{\lambda_i h}(\v x^i - \lambda_i \nabla f(\v x^i)) - \operatorname{prox}_{\lambda_i h}(\v x^i - \lambda_i \v g^i) \| \\
			&\leq \| \nabla f(\v x^i) - \v g^i \|. 
		\end{aligned}
		\label{equ:firmly}
	\end{equation*}
	The desired conclusion follows immediately from the triangle inequality and the definition of $\varepsilon_i^{\v g}$ in \eqref{equ:def_eg}. 
\end{proof}
Based on Lemma \ref{lem:KKT} and the condition \eqref{equ:general_gradient_inexact_new} with $\lambda = \lambda_i$ and noting $\v x^{i+1} = \v x^i(\lambda_i)$, we introduce the following termination criterion: 
\begin{equation}
	\begin{aligned}
		\label{equ:termination}
		\Delta_i^{\v g} :={} &\|\v{G}_{\lambda_i}(\v{x}^i, \v{g}^i)\| \\
		& +  \sqrt{(\eta_i^{\v{g}})^2 + \left( a^2 \lambda_i^{-2} + b^2 \right) \|\v x^{i+1} - \v{x}^i\|^2} \leq \epsilon.  
	\end{aligned}
\end{equation}
With this, we are ready to present the pseudo-codes of the APIG framework, outlined in Algorithm \ref{frm:PIG}. 

\begin{algorithm}[t]
	\caption{Proposed APIG Framework}\label{frm:PIG}
	\begin{algorithmic}[1]
		\STATE Initialize: $\v x^0 \in \mathbb{R}^n$, $\epsilon > 0$, $\lambda_{0}^{(0)}, \lambda_{1}^{(0)} \in [\lambda_{\min}, \lambda_{\max}]$, $\alpha \in (0, 1)$. Set $i = 0$. 
		\FOR{$i = 0, 1, \ldots$}
		\STATE Find the smallest nonnegative integer $\ell$ such that $\lambda_i = \alpha^{\ell} \lambda_{i}^{(0)}$ satisfies the LS condition \eqref{equ:B1} or \eqref{equ:B2}, under the inexactness conditions \eqref{equ:general_gradient_inexact_new}, \eqref{equ:function_error1}, and \eqref{equ:function_error2}. 
		\STATE Update $\v x^{i+1} = \operatorname{prox}_{\lambda_i h}(\v x^{i} - \lambda_i \v g^i)$.
		\STATE \textbf{If} \eqref{equ:termination} holds, \textbf{then} return $\v x^i$ and break. 
		\STATE \revxl{Update $\lambda_{i+1}^{(0)}$ via \eqref{equ:BB} when $i \geq 1$.}
		\ENDFOR
	\end{algorithmic}
\end{algorithm}

The following result characterizes the controlled descent property of the sequence \(\{F(\v{x}^i)\}\).
\begin{lemma}[Controlled Descent Property]\label{lem:descent_property}
	Under the same assumptions and parameter settings as in Lemma \ref{lem:ls}, let \( \{\mathbf{x}^i\} \) be the sequence generated by Algorithm \ref{frm:PIG}. Then the following descent guarantees hold: 
	\begin{itemize}
		\item[(i)] If the LS condition \eqref{equ:B1} is used, then for all $i \geq 0$, we have 
		\begin{equation}
			F(\v x^{i+1}) - F(\v x^i) \leq - \frac{\theta \lambda_i}{2}  \|\v G_{\lambda_i}(\v x^i, \v g^i)\|^2 + \nu_i + 2\eta_i^f,
			\label{equ:descent_property}
		\end{equation}
		where $\nu_i$ and $\eta_i^f$ are defined in \eqref{equ:def_nu} and \eqref{equ:def_etaif}, respectively. 
		\item[(ii)] If the LS condition \eqref{equ:B2} is used, then the inequality \eqref{equ:descent_property} holds with $\theta = 1/2$.  
	\end{itemize}
\end{lemma}
\begin{proof}
	See Appendix~\ref{apd:finite_termination} in the Supplementary Material.  
\end{proof}

Based on Lemma \ref{lem:descent_property}, we now establish the convergence results for the proposed  APIG framework. 
\begin{theorem}[Iteration Complexity and Convergence Rate]
	\label{thm:PIG_convergence}
	Consider the same assumptions and settings as in Lemma \ref{lem:ls}. 
	Then, Algorithm \ref{frm:PIG} terminates within $N_s = \mathcal{O}(\epsilon^{-2})$ iterations, and $\v x^{N_s}$ is an $\epsilon$-stationary point of problem \eqref{equ:generic_problem}. 
	\label{thm:1}
	For convex $f(\cdot)$, the following results hold: 
	\begin{itemize}
		\item \textit{Case I:} Under the LS condition \eqref{equ:B2} with summable \( \{\eta_i^{\v g}\} \) and \( \{\eta_i^f\} \), \( a = b = 0 \), and \( c \in [0, \theta/4] \), the average iterate $\overline{\v x}^N = (N+1)^{-1} \sum_{i = 0}^N \v x^i$ with \( 0 \leq N \leq N_s \) satisfies 
		\[
		F(\overline{\v x}^N) - F(\v{x}^\star) \leq \mathcal{O}\left( (N+1)^{-1} \right). 
		\]
		\item \textit{Case II:} Under either the LS condition \eqref{equ:B1} or \eqref{equ:B2} with \( \eta_i^{\v g} \equiv 0 \), \( a \in [0, 1) \), \( b \in [0, \infty) \), \( \eta_i^f \equiv 0 \), and \( c \in [0, \theta/4] \),  for \( 2 \leq N \leq N_s \), we have
		\begin{equation*}
			\begin{aligned}
				& F(\v{x}^N) - F(\v{x}^\star) \\
				\leq{} & \max \left\{ \mathcal{O}\left( \sqrt{2}^{-(N-1)} \right), \mathcal{O}\left( (N-1)^{-1} \right) \right\}.
			\end{aligned}
		\end{equation*}
	\end{itemize}
\end{theorem}
\begin{proof}
	See Appendix~\ref{apd:PIG_convergence}. 
\end{proof}

Theorem \ref{thm:1} extends the results of \cite{schmidt2011ConvergenceRatesInexact} and \cite{paquette2020StochasticLine} in several aspects within the context of the PIG method \eqref{equ:PIG}. 
First, in the convex case, it significantly improves upon the results in \cite{schmidt2011ConvergenceRatesInexact} by explicitly accounting for function errors, and achieves the same iteration complexity as those in \cite{paquette2020StochasticLine}. 
Second, in the nonconvex case, compared to \cite{paquette2020StochasticLine}, Theorem \ref{thm:1} achieves the same complexity result as \cite{paquette2020StochasticLine}, but under a much more general framework. 
Specifically, our proposed framework incorporates the more general inexactness condition \eqref{equ:general_gradient_inexact_new}, \eqref{equ:function_error1}, and \eqref{equ:function_error2}, and allows for a more general convex $h(\cdot)$ in problem \eqref{equ:generic_problem}.

\section{JBCP Application}
In this section, we custom-apply the proposed APIG framework to an important signal processing problem in cooperative cellular communication networks{\textemdash}the JBCP with PAPCs. 
More specifically, in Section \ref{ss:3a}, we first introduce and formulate the JBCP with PAPCs, transform it into an equivalent semidefinite program (SDP), and derive its Lagrange dual problem, which takes the form of problem \eqref{equ:generic_problem}. 
In Sections \ref{ss:3b} and \ref{ss:3c}, we provide a way of computing the function and gradient information of the dual problem and their approximations through solving a weighted JBCP without PAPCs by an efficient FP iteration algorithm, respectively.
In Section \ref{ss:3d}, we then apply the proposed APIG framework to solve the dual problem, leading to an efficient tailored algorithm for solving the JBCP with PAPCs, named APIG-FP. 
Finally, we present numerical results in Section \ref{ss:3e} to demonstrate  the effectiveness of the proposed algorithm. 

\subsection{Problem Formulation, SDR, and Lagrangian Dual}\label{ss:3a}
\subsubsection{System Model and Problem Formulation} 
In cooperative cellular networks, multiple relay-like BSs are connected to a central processor (CP) via fronthaul links with limited capacities, allowing joint processing at the CP to mitigate intercell interference by sharing user data among BSs. However, such cooperation imposes heavy demands on fronthaul links. 
To alleviate this, strategies that jointly design BS transmissions and fronthaul utilization have been proposed \cite{
	dai2014SparseBeamformingUsercentric, 
	shi2014GroupSparseBeamforming,
	park2013JointPrecodingMultivariate, 
	park2014InterclusterDesignPrecoding,
	patil2014HybridCompressionMessagesharing, 
	kang2016FronthaulCompressionPrecoding,
	zhou2016FronthaulCompressionTransmit,
	he2019HybridPrecoderDesign, 
	kim2019JointDesignFronthauling,
	ahn2020FronthaulCompressionPrecoding,                                                    
	liu2021UplinkdownlinkDualityMultipleaccess,
	fan2022EfficientlyGloballySolving,
	fan2023QoSbasedBeamformingCompression
}. 
Recent works have addressed the JBCP by minimizing the total transmit power while ensuring users' signal-to-interference-and-noise ratio (SINR) requirements and BSs' fronthaul rate constraints \cite{liu2021UplinkdownlinkDualityMultipleaccess, fan2022EfficientlyGloballySolving, fan2023QoSbasedBeamformingCompression}. 
In this paper, we solve the JBCP with the more practical PAPCs, as considered in  \cite{fan2024JointBeamformingCompression}.

We adopt the same system model as in \cite{fan2024JointBeamformingCompression}. 
Consider a cooperative cellular network consisting of one CP and $M$ single-antenna BSs, which are connected to the CP through the noiseless fronthaul links with limited capacities. 
These BSs cooperatively serve $K$ single-antenna users via a noisy wireless channel. 
Let $\cM = \{1,2,\d,M\}$ and $\cK=\{1,2,\d,K\}$ represent the sets of BSs and users, respectively. 
Let $\v{v}_k = [v_{k,1}, v_{k,2}, \d , v_{k,M}]^\transpose$ be the $M \times 1$ beamforming vector, 
$\m Q \in \mathbb{C}^{M\times M}$ be the covariance matrix of the additive compression noise at the BSs, 
$\sigma_k^2$ be the noise power at user~$k$,  and $\v{h}_k = [h_{k,1}, h_{k,2}, \d, h_{k,M}]^\hermitian$ be the channel vector of user $k$. 
The transmit power of BS/antenna $m$ is given by
$$
\operatorname{PW}_m = \sum_{k\in\cK} |v_{k,m}|^2 + \m Q^{(m,m)},
$$ and the SINR of user $k$ is
\begin{equation*}
	\operatorname{SINR}_k = \frac{|\v{h}_k^\hermitian \v{v}_k|^2}{\sum_{j\neq k} |\v{h}_k^\hermitian \v{v}_j|^2 + \v{h}_k^\hermitian \m{Q} \v{h}_k + \sigma_k^2},~\forall\,k\in\cK. 
\end{equation*}

To fully utilize the fronthaul links with limited capacities, we adopt the information-theoretically optimal multivariate compression strategy \cite{park2013JointPrecodingMultivariate} to compress the signals from the CP to the BSs. 
Without loss of generality, we assume that the compression order is from BS $ M $ to BS $ 1 $. 
Then, the fronthaul rate of BS $m$ is given by 
\begin{equation*}
	\begin{aligned}
		\operatorname{FR}_m = \log_2 \left(\frac{\sum_{k\in\cK} |v_{k,m}|^2 + \m{Q}^{(m,m)}}{ \m{Q}^{(m:M, m:M)}/\m{Q}^{(m+1:M, m+1:M)} }\right),\fm. 
	\end{aligned}
\end{equation*}
Here, $\m{Q}^{(m:M, m:M)}$ denotes the principal submatrix of $\m Q$ formed by the rows and columns indexed by  indices $\{m, m+1, \d, M\}$, and $\m{Q}^{(m:M, m:M)}/\m{Q}^{(m+1:M, m+1:M)}$ represents the Schur complement of the block $\m{Q}^{(m+1:M, m+1:M)}$ of $\m{Q}^{(m:M, m:M)}$. 

Given a set of SINR targets for the users $\{\overline\gamma_k\}$, a set of fronthaul capacities for the BSs $\{\overline{C}_m\}$, and a set of per-antenna (i.e., per-BS) power budgets for the BSs $\{\overline{P}_m\}$, 
we aim to minimize the total transmit power of all BSs while satisfying all users' SINR constraints, all BSs' fronthaul rate constraints, and all PAPCs. 
The problem can be formulated as follows: 
\begin{equation}
	\begin{aligned}
		\min_{\{\v v_k\}, \m{Q}\succeq \m{0}} &\quad \sum_{k\in\cK} \|\v{v}_k\|^2 + \tr(\m Q)\\
		\st~~~ &\quad \operatorname{SINR}_k \geq \bar\gamma_k, \fk{},\\
		&\quad \operatorname{FR}_m \leq \overline{C}_m, ~\operatorname{PW}_m \leq \overline{P}_m, \fm.
	\end{aligned}
	\label{equ:JBCP_PAPC}
\end{equation}
The last constraint in problem \eqref{equ:JBCP_PAPC} is the PAPC of antenna~$m$ (i.e., BS $m$). 
While the PAPC is important in practice (as each antenna has its own power budget), 
it introduces a technical challenge in solving problem \eqref{equ:JBCP_PAPC}. 
Specifically, without the PAPCs, problem \eqref{equ:JBCP_PAPC} can be solved efficiently and globally via solving two FP equations \cite{fan2022EfficientlyGloballySolving}. 
However, the existence of PAPCs makes the algorithm proposed in \cite{fan2022EfficientlyGloballySolving} inapplicable. 

\subsubsection{Equivalent SDP Reformulation of Problem \eqref{equ:JBCP_PAPC}}
By using the similar arguments as in \cite[Proposition 4]{liu2021UplinkdownlinkDualityMultipleaccess}, 
we can equivalently rewrite the users' SINR constraints and the BSs' fronthaul rate constraints as
\begin{equation}\label{equ:jbcp:c1}
	(1 + \overline{\gamma}_k^{-1})|\v{v}_k^\hermitian\v h_k|^2 - \sum_{j\in\cK} |\v{v}_j^\hermitian \v h_k|^2 - \v h_k^\hermitian \m Q \v h_k \geq \sigma_k^2, \fk
\end{equation}
and 
\begin{equation}\label{equ:jbcp:c2}
	2^{\overline{C}_m} \begin{bmatrix}
		\m{0}& \m{0} \\
		\m{0}& \m{Q}^{(m:M, m:M)}
	\end{bmatrix} - \operatorname{PW}_m \cdot~ \m{E}_m \succeq \m{0}, \fm.
\end{equation}
Then, we can  reformulate problem \eqref{equ:JBCP_PAPC} as
\begin{equation}
	\label{equ:P}
	\begin{aligned}
		\min_{\{\v v_k\}, \m{Q}\succeq \m{0}}&\quad \sum_{k\in\cK} \|\v{v}_k\|^2 + \tr(\m Q)\\
		\st~~~~ &\quad  \eqref{equ:jbcp:c1},~\eqref{equ:jbcp:c2},~\text{and}~ \operatorname{PW}_m \leq \overline{P}_m, \fm. 
	\end{aligned} 
	\end{equation}
Problem \eqref{equ:P} is a (nonconvex) quadratically constrained quadratic program (QCQP). A well-known technique to tackle the QCQP is SDR \cite{luo2010SemidefiniteRelaxationQuadratic, xu2023new}. 
By applying the SDR technique (with $\m V_k = \v v_k \v v_k^\hermitian$) to problem \eqref{equ:P}, we obtain its SDR formulation as follows:
\begin{subequations}
	{\label{equ:JBCP_PAPC_SDR}\small
		\begin{align}
			\min_{\substack{{\m V_k \succeq \m 0}, \\ \m Q \succeq \m 0}} & ~ \sum_{k \in \cK} \tr(\m V_k) + \tr(\m Q) \tag{\ref{equ:JBCP_PAPC_SDR}} \\
			\st~ &  \left \langle (1 + \overline{\gamma}_k^{-1}) \m V_k - \sum_{j \in \cK} \m V_j - \m Q,  \v h_k \v h_k^\hermitian\right \rangle \geq \sigma_k^2, \fk, \label{cst:SINR}\\
			&  \begin{bmatrix}
				\m 0 & \m 0 \\
				\m 0 & \m Q^{(m:M,m:M)}
			\end{bmatrix}  -   \frac{\sum_{k\in \cK} \m V_k^{(m,m)} + \m Q^{(m,m)}}{2^{\overline{C}_m}} \m E_m \succeq \m 0, \label{cst:FR}\\
			& \sum_{k\in \cK} \m V_k^{(m,m)} + \m Q^{(m,m)} \leq \overline{P}_m, \fm. \label{cst:PAPC}
	\end{align}}
\end{subequations} 
\!\!\!We have the following proposition, whose proof is similar to that of \cite[Theorem 1]{fan2022EfficientlyGloballySolving}, and is therefore omitted here. 
\begin{proposition}
	\label{prop:tight}
	If problem \eqref{equ:JBCP_PAPC_SDR} is strictly feasible, then its optimal solution $(\{\m V_k^*\}, \m Q^*)$ always satisfies $\rk(\m V_k^*) = 1$ for all $k\in\cK$. 
\end{proposition}

This proposition shows that problem \eqref{equ:JBCP_PAPC_SDR} always has a rank-one solution for $\left\{\m V_k^\star \right\}$, ensuring that the SDR of problem \eqref{equ:P} is tight. 
Consequently, the SDP problem \eqref{equ:JBCP_PAPC_SDR} is an equivalent reformulation of problem \eqref{equ:P}, allowing us to solve the latter via addressing the former. 
Instead of directly using a solver (e.g., CVX \cite{CVX}) to solve problem \eqref{equ:JBCP_PAPC_SDR} (due to its high computational cost), we further reformulate it into the form of problem \eqref{equ:generic_problem}, and apply the proposed APIG framework to solve the reformulation. 
In the following of this paper, we assume the strict feasibility of problem \eqref{equ:JBCP_PAPC_SDR}. 
  
\subsubsection{Lagrangian Dual of Problem \eqref{equ:JBCP_PAPC_SDR}}
When constraint \eqref{cst:PAPC} is removed, problem \eqref{equ:JBCP_PAPC_SDR} can be efficiently solved using the FP iteration algorithm proposed in  \cite{fan2022EfficientlyGloballySolving}. 
To leverage this algorithm, we consider the Lagrange dual problem of problem \eqref{equ:JBCP_PAPC_SDR}. 
Let $\v x = [x_1, x_2, \d, x_M]^\transpose \geq \v 0$ be the Lagrange multiplier associated with the inequality constraints in \eqref{cst:PAPC}. 
The dual problem is
\begin{equation}\label{equ:JBCP_PAPC_dual}
	\max_{\v x \geq \v 0} \  d(\v x), 
\end{equation}
where $d(\v x)$ is the optimal value of the following problem: 
\begin{equation}\small 
	\label{equ:JBCP_PAPC_dual_sub}
	\begin{aligned}
		\min_{\substack{\{\m V_k \succeq \m 0\} \\ \m{Q}\succeq \m{0}} } &~ \sum_{m \in \cM} \left((1+x_m) \left( \sum_{k\in\cK} \m V_k^{(m,m)} + \m Q^{(m,m)} \right) - x_m\overline{P}_m\right)\\
		\st~~~ &~ \eqref{cst:SINR} \text{~and~} \eqref{cst:FR}. 
	\end{aligned}
\end{equation} 
According to the classical duality theory \cite[p. 216]{boyd2004ConvexOptimization}, the function $d(\cdot)$ in \eqref{equ:JBCP_PAPC_dual} is concave and typically  nondifferentiable. 
Hence, the subgradient algorithm is commonly used to solve the convex problem \eqref{equ:JBCP_PAPC_dual} (e.g., as in \cite{yu2007TransmitterOptimizationMultiantenna, dartmann2013DualityMaxminBeamforming, zhang2020DeepLearningEnabled, miretti2024ULDLDualityCellfree}). 
However, it has been shown recently in \cite{fan2023QoSbasedBeamformingCompression} that problem \eqref{equ:JBCP_PAPC_dual_sub} has a unique solution for any $\v x \geq \v 0$, 
which implies that $d(\cdot)$ is differentiable by Danskin's Theorem \cite{hiriart-urruty1996ConvexAnalysisMinimization}, enabling the use of the PG method \eqref{equ:PG} to solve problem \eqref{equ:JBCP_PAPC_dual}. 
The gradient of $d(\cdot)$ is given in the following proposition. 
\begin{proposition}
	\label{prop:diff}
	The objective function $d(\cdot)$ in problem \eqref{equ:JBCP_PAPC_dual} is differentiable on $\mathbb{R}_+^M$. 
	For $\v x \geq \v 0$, let $\left(\{\m V_k^\star(\v x)\}, \m Q^\star(\v x\right))$ be the unique solution to problem~\eqref{equ:JBCP_PAPC_dual_sub}.  Then, the $m$-th component of the gradient $\nabla d(\v x)$ with $m \in \mathcal{M}$ is given by
	\begin{equation}
		\label{equ:grad}
		(\nabla d(\v x))_m = \sum_{k\in\cK} \m V_k^{\star}(\v x)^{(m,m)} + \m Q^{\star}(\v x)^{(m,m)}  - \overline{P}_m. 
	\end{equation}
\end{proposition}

\subsection{Compute $d(\v x)$ and $\nabla d(\v x)$}\label{ss:3b}

Proposition~\ref{prop:diff} shows that computing $\nabla d(\mathbf{x})$ at any given point $\mathbf{x} \geq \mathbf{0}$ requires solving problem \eqref{equ:JBCP_PAPC_dual_sub} to global optimality. 
Specifically, problem \eqref{equ:JBCP_PAPC_dual_sub} is a weighted total transmit power minimization problem subject to all users' SINR constraints and all BSs' fronthaul rate constraints, which can be solved globally by the FP iteration algorithm proposed in \cite{fan2022EfficientlyGloballySolving}. 
Below, we briefly describe this algorithm. 

The first stage of the algorithm is to solve the following FP equations:
\begin{subequations}\label{equ:fp_equation}
		\begin{align}
			\vv \beta &= I_{\v x}(\vv \beta), \label{equ:dual_fp}\\
			\v p &= J_{\vv \beta, \v x}(\v p), \label{equ:primal_fp}
		\end{align}
\end{subequations} 
where $\vv \beta, \v p \in \mathbb{R}^K$. 
Here, the mappings $I_{\v x}: \mathbb{R}^K \to \mathbb{R}_{++}^K$ and $J_{\vv \beta, \v x}: \mathbb{R}^K \to \mathbb{R}_{++}^K$ are defined componentwise\footnote{It is worth noting that the mappings $I_{\v x}(\cdot)$, $J_{\vv \beta, \v x}(\cdot)$, and the forthcoming $\m Q_{\vv \beta, \v x}(\cdot)$ follow the forms as in \cite{fan2022EfficientlyGloballySolving} and \cite{fan2023QoSbasedBeamformingCompression}, but with explicit incorporation of the dependence of $\v x$ and $\vv \beta$.}.
Specifically, the $k$-th component of $I_{\v x}(\vv \beta)$ with $k \in \mathcal{K}$ is given by 
	\begin{equation}\label{equ:def_I}
		(I_{\v x}(\vv \beta))_k = \frac{\overline{\gamma}_k}{\overline{\gamma}_k + 1} \cdot \frac{1}{\v h_k^\hermitian \m C(\vv \beta, \{\m \Lambda_{\v x, m}(\vv \beta)\}, \v x)^{-1} \v h_k}, 
	\end{equation}
	where 
	\begin{multline*}
		\m C(\vv \beta, \{\m \Lambda_{\v x, m}(\vv \beta)\}, \v x) = \m I + \sum_{k\in\cK} \beta_k \v h_k \v h_k^\hermitian + \diag(\v x) \\ + \diag(\m \Lambda_{\v x, 1}(\vv \beta)^{(1,1)}, \m \Lambda_{\v x, 2}(\vv \beta)^{(2,2)}, \d, \m \Lambda_{\v x, M}(\vv \beta)^{(M,M)}) 
	\end{multline*}
	and $\{\m \Lambda_{\v x, m}(\cdot)\}$ are the mappings given in \cite[Section 3.2.1]{fan2022EfficientlyGloballySolving}. 
	The $k$-th component of $J_{\vv \beta, \v x}(\v p)$ with $k \in \cK$ is defined as 
	\begin{equation}\small 
		\begin{aligned}
			&(J_{\vv \beta, \v x}(\v p))_k  \\
			={}& \frac{\gamma_k \left( \sum_{j\neq k} p_j |\v h_k^\hermitian \v u_j(\vv \beta, \v x)|^2 + \v h_k^\hermitian \m Q_{\vv \beta, \v x}(\v p) \v h_k + \sigma_k^2\right)}{|\v h_k^\hermitian \v u_k(\vv \beta, \v x)|^2}, 
		\end{aligned}
		\label{equ:def_Jk}
	\end{equation}
	where 
	\begin{equation}
		\begin{aligned}
			\v u_k(\vv \beta, \v x) &= \frac{\m C(\vv \beta, \{\m \Lambda_{\v x, m}(\vv \beta)\}, \v x)^{-1} \v h_k}{\|\m C(\vv \beta, \{\m \Lambda_{\v x, m}(\vv \beta)\}, \v x)^{-1} \v h_k\|}
		\end{aligned} \label{equ:def_u}
	\end{equation}
	and $\m Q_{\vv \beta, \v x}(\cdot)$ is given in the paragraph under Eq. (29) in \cite{fan2023QoSbasedBeamformingCompression}. 
	It is worth mentioning that $\m Q_{\vv \beta, \v x}(\v p)$ has a special linear relationship with respect to $\v p$, given by 
	\begin{equation}
		\m Q_{\vv \beta, \v x}(\v p) = \sum_{k\in\cK} p_k \underline{\m A}_k(\vv \beta, \v x), 
		\label{equ:def_Q}
	\end{equation}
	where $\underline{\m A}_k(\vv \beta, \v x) \succeq \m 0$. 
	
	The second stage of the algorithm is to compute the function value and the gradient based on the solutions to the equations in \eqref{equ:fp_equation}.
	Let $\vv \beta^\star(\v x)$ and 
	\begin{equation}
		\v p^\star(\v x) = \v p^\star(\vv \beta^\star(\v x), \v x)
		\label{equ:pstarx}
	\end{equation}
	be the unique solutions to the equations in \eqref{equ:fp_equation} (such uniqueness is guaranteed by Lemma \ref{lem:fixedpoint} in Appendix \ref{apd:lemmas} in the Supplementary Material). 
	The solutions $\{\m V_k^\star(\v x)\}$ and $\m Q^\star(\v x)$ to problem~\eqref{equ:JBCP_PAPC_dual_sub} can be recovered by
	\begin{equation*}
		\begin{aligned}
			\m V_k^\star(\v x) &= (\v p^\star(\v x))_k \v u_k(\vv \beta^\star(\v x), \v x) \v u_k(\vv \beta^\star(\v x), \v x)^\hermitian, \fk,\\
			\m Q^\star(\v x) &= \m Q_{\vv \beta^\star(\v x), \v x}(\v p^\star(\v x)). 
		\end{aligned}
		\label{equ:reconstruction}
	\end{equation*}
	Then, by \eqref{equ:grad}, we can compute $\nabla d(\v x)$ using $\vv \beta^\star (\v x)$ and $\v p^\star(\v x)$ as follows: 
	\begin{equation}
		\nabla d(\v x) = \v g(\vv \beta^\star(\v x), \v p^\star(\v x), \v x),
		\label{equ:grad2}
	\end{equation}
	where the $m$-th component of $\v g(\cdot, \cdot, \cdot)$ is defined as
	\begin{equation}\label{equ:compute_ap_grad_component}\small 
		(\v g(\vv \beta, \v p, \v x))_m = \sum_{k \in \cK} p_k |\v u_k(\vv \beta, \v x)^{(m)}|^2 + \m Q_{\vv \beta, \v x}(\v p) - \overline{P}_m. 
	\end{equation}
	Furthermore, we can compute $d(\v x)$ using the optimal value of the dual problem of problem \eqref{equ:JBCP_PAPC_dual_sub} \cite{fan2023QoSbasedBeamformingCompression}, given by
\begin{equation}
		d(\v x) =  \widetilde{d}(\vv \beta^\star(\v x), \v x) := \sum_{k\in\cK} \beta_k^\star(\v x) \sigma_k^2 - \sum_{m\in\cM} x_m \overline{P}_m. 
		\label{equ:func}
\end{equation}

However, although the FP iteration algorithm in \cite{fan2022EfficientlyGloballySolving} can solve problem \eqref{equ:JBCP_PAPC_dual_sub} to global optimality, the computational cost of finding the global solution or a very high-precision solution generally is high. Consequently, directly applying the PG method \eqref{equ:PG} to solve problem \eqref{equ:JBCP_PAPC_dual} maybe inefficient. To address this, we propose leveraging the proposed APIG framework to solve problem \eqref{equ:JBCP_PAPC_dual}, which involves utilizing an inexact solution to problem \eqref{equ:JBCP_PAPC_dual_sub} and computing approximations of $d(\v x)$ and $\nabla d(\v x)$. The next subsection details this approach.

\subsection{Compute Approximations of $d(\v x)$ and $\nabla d(\v x)$}\label{ss:3c}

In this subsection, we first outline the FP iteration algorithm for solving problem \eqref{equ:JBCP_PAPC_dual_sub} inexactly for a given $\v x \geq \v 0$. 
Then, we describe how to use \eqref{equ:grad2} and \eqref{equ:func} with inexact solutions to compute approximations of $d(\v x)$ and $\nabla d(\v x)$. 
Finally, we demonstrate that the errors in these approximations are bounded by a factor proportional to the precision of the FP iteration algorithm. 

Given two stopping parameters $\operatorname{res}_1 \geq 0, \operatorname{res}_2\geq 0$,  define $\operatorname{res} =  \operatorname{res}_1 +  \operatorname{res}_2$. 
The FP iteration algorithm proposed in \cite{fan2022EfficientlyGloballySolving} solves the FP equations \eqref{equ:fp_equation} as follows:

Step 1: Starting from $\vv \beta^{(0)} \in \mathbb{R}_+^K$, the FP equation $I_{\v x}(\cdot)$ is solved iteratively as
\begin{equation}
	\begin{aligned}
		\vv \beta^{(i+1)} &= I_{\v x}(\vv \beta^{(i)}),\quad \forall\, i=0,1,\d. 
	\end{aligned}
	\label{equ:dual_iter}
\end{equation}
The iteration terminates once the FP residual  satisfies
\begin{equation}
	\mu(\vv \beta^{(i)}, \vv \beta^{(i+1)})\leq \operatorname{res}_1, 
	\label{equ:res1}
\end{equation}
where $\mu(\cdot, \cdot)$ is Thompson's metric defined in \eqref{equ:metric}. 
The approximate FP of $I_{\v x}(\cdot)$ is denoted as $\widetilde{\vv \beta}: = \vv \beta^{(i+1)}$. 

Step 2: Starting from $\v p^{(0)} \in \mathbb{R}_+^K$, the FP equation $J_{\widetilde{\vv \beta}, \v x}(\cdot)$ is solved iteratively as
\begin{equation}
	\begin{aligned}
		\v p^{(j+1)} &= J_{\widetilde{\vv \beta}, \v x}(\v p^{(j)}),\quad\forall\, j=0,1,\d. 
	\end{aligned}
	\label{equ:primal_iter}
\end{equation}
The iteration terminates when 
\begin{equation}
	\mu(\v p^{(j)}, \v p^{(j+1)}) \leq \operatorname{res}_2. 
	\label{equ:res2}
\end{equation}
The approximate FP of $J_{\widetilde{\vv \beta}, \v x}(\cdot)$ is denoted as $\widetilde{\v p}: = \v p^{(j+1)}$.

After obtaining the approximate FPs $\widetilde{\vv \beta}$ and $\widetilde{\v p}$, the inexact function and gradient are computed by replacing $\vv \beta^\star(\v x)$ and $\v p^\star(\v x)$ with $\widetilde{\vv \beta}$ and $\widetilde{\v p}$ in \eqref{equ:grad2} and \eqref{equ:func}, respectively. 
They are denoted by $\widetilde{d}(\widetilde{\vv \beta}, \v x)$ and $\v g(\widetilde{\vv \beta}, \widetilde{\v p}, \v x)$. 
To effectively control the terms $| \widetilde{d}(\widetilde{\vv \beta}, \v x) - d(\v x) |$ and $\| \v g(\widetilde{\vv \beta}, \widetilde{\v p}, \v x) - \nabla d(\v x)\|$ via the stopping parameter $\operatorname{res}$, we introduce the following assumption. 
\begin{assumption}\label{ass:lip}
\begin{itemize}
	\item[(a)] The FP equation \eqref{equ:primal_fp} with $\vv \beta = \widetilde{\vv \beta}$ has a FP. 
	\item[(b)] There exists a constant $\kappa_1(\v x)>0$ such that 
	\begin{equation}
		\mu(J_{\widetilde{\vv \beta}, \v x}(\widetilde{\v p}), J_{\widetilde{\vv \beta}, \v x}(\v p^\star(\v x))) \leq \kappa_1(\v x) \mu(\widetilde{\v p}, \v p^\star(\v x)). 
		\label{equ:lc_primal_ass}
	\end{equation}
	\item[(c)] The following Lipschitz continuity conditions hold for $\v p^\star(\cdot, \cdot)$ and $\v g(\cdot, \cdot, \cdot)$: 
	\begin{equation}
		\|\v p^\star(\widetilde{\vv \beta}, \v x) - \v p^\star(\vv \beta^\star(\v x), \v x)\| \leq L_1(\v x) \|\widetilde{\vv \beta} - \vv \beta^\star(\v x)\|
		\label{equ:asss1}
	\end{equation}
	and 
	\begin{equation}
		\begin{aligned}
			&\|\v g(\widetilde{\vv \beta}, \widetilde{\v p}, \v x) - \v g(\vv \beta^\star(\v x), \v p^\star(\v x), \v x)\|\\
			\leq{}& L_2(\v x) (\|\widetilde{\vv \beta} - \vv \beta^\star(\v x)\| + \| \widetilde{\v p} - \v p^\star(\v x) \|),
		\end{aligned}
		\label{equ:asss2}
	\end{equation}
	where $L_1(\v x) \geq 0$ and $L_2(\v x) \geq 0$ are constants. Note that $\v p^\star(\v x) = \v p^\star(\vv \beta^\star(\v x), \v x)$ as defined in \eqref{equ:pstarx}. 
\end{itemize}
\end{assumption}

As shown in \cite{fan2022EfficientlyGloballySolving}, when \(\operatorname{res}_1 = 0\), \(\widetilde{\vv \beta}\) is the exact FP \(\vv \beta^\star(\v x)\), and thus solves the FP equation \eqref{equ:dual_fp}. 
In this case, \eqref{equ:lc_primal_ass} holds because of \eqref{equ:lc3} in Lemma \ref{lem:properties} (c) in Appendix \ref{apd:lemmas} in the Supplementary Material,
\eqref{equ:asss1} holds naturally, and \eqref{equ:asss2} also holds because the mapping $\v g(\vv \beta, \v p, \v x)$ depends linearly on $\v p$, as shown in \eqref{equ:def_Q} and \eqref{equ:compute_ap_grad_component}. 
Therefore, Assumption \ref{ass:lip} is satisfied. 
We assume similar conditions hold for general choices of \(\operatorname{res}_1\), specifically for sufficiently small \(\operatorname{res}_1\) to ensure that \(\widetilde{\vv \beta}\) is sufficiently close to \(\vv \beta^\star\). 

Under Assumption \ref{ass:lip}, we have the following result. 
\begin{proposition}
	Let $(\widetilde{\vv \beta}, \widetilde{\v p})$ be the output of the iterative procedure \eqref{equ:dual_iter}--\eqref{equ:res2} with termination parameter $\operatorname{res}$. 
	If Assumption \ref{ass:lip} holds, then there exists a constant $C(\v x) \geq 0$ such that  
	\begin{subequations}\label{equ:control}
		\begin{align}
			| \widetilde{d}(\widetilde{\vv \beta}, \v x) - d(\v x) |	 \leq C(\v x) \operatorname{res}, \label{equ:control_function}\\
			\| \v g(\widetilde{\vv \beta}, \widetilde{\v p}, \v x) - \nabla d(\v x)\| \leq C(\v x) \operatorname{res}\!. 
			\label{equ:control_gradient}
		\end{align}
	\end{subequations}
	\label{prop:error_control}
\end{proposition}
\begin{proof}
	See Appendix~\ref{apd:error_control} in the Supplementary Material. 
\end{proof}

\subsection{A Practical Algorithm for Solving Problem \eqref{equ:JBCP_PAPC_dual}}\label{ss:3d}

\begin{algorithm}[t]
	\caption{Proposed APIG-FP Algorithm}\label{alg:APIG-FP}
	\begin{algorithmic}[1]
		\STATE Initialize: summable nonnegative sequences $\{(\widetilde{\eta}_i^{\v g})^2\}$ and $\{\widetilde{\eta}_i^f\}, \widetilde{b} \geq 0$, $\theta \in (0,1)$, $\varrho > 1$, $\epsilon > 0$, $\lambda_0^{(0)}, \lambda_1^{(1)} \in [\lambda_{\min}, \lambda_{\max}]$, $\alpha \in (0, 1)$, $\v x^0 \in \mathbb{R}_+^n$. Set $i = 0$ and give an initial guess of $C_0$ and  $\widetilde C_0$. 
		
		\FOR{$i = 0,1,\d $} 
		\FOR{$\ell = 0,1, \ldots$}
		\STATE Set $\lambda = \lambda_{i}^{(0)} \alpha^\ell$.  
		\STATE Use the iterative procedure \eqref{equ:dual_iter}--\eqref{equ:res2} with $\operatorname{res}$ given in \eqref{equ:set_res} to solve problem~\eqref{equ:JBCP_PAPC_dual_sub} with $\v x = \v x^i$ and $\v x = \v x^i(\lambda)$, respectively, to obtain $(f_i, f_i(\lambda), \v g^i)$.  
		\STATE \textbf{If} \eqref{equ:B1:JBCP} holds \textbf{then} break.
		\STATE Update $C_i = \varrho C_i$ and  $\widetilde C_i  =  \varrho \widetilde C_i$. 
		\ENDFOR
		
		\STATE Set $\lambda_i = \lambda$ and update $\v x^{i+1} = \v x^{i}(\lambda)$.
		
		\STATE \textbf{If} \eqref{equ:termination2} holds \textbf{then} return $\v x^i$ and break. 
		
		\STATE \revxl{Set $\lambda_{i+1}^{(0)}$ via \eqref{equ:BB} when $i \geq 1$}, $C_{i+1} = C_i$, and $\widetilde{C}_{i+1} = \widetilde{C}_i$.
		\ENDFOR
	\end{algorithmic}
\end{algorithm}

Now it is evident that problem~\eqref{equ:JBCP_PAPC_dual}, an equivalent reformulation of problem \eqref{equ:JBCP_PAPC}, is a special instance of problem~\eqref{equ:generic_problem} with 
$f(\cdot) = -d(\cdot)$ and $h(\cdot)$ being the indicator function of $\mathbb{R}_{+}^M$.  
Therefore, we can apply Algorithm \ref{frm:PIG} to solve problem \eqref{equ:JBCP_PAPC_dual}. To achieve this, we need to compute the inexact function and gradient information at $\v x^i$ and $\v x^i(\lambda)$ such that the inexact conditions \eqref{equ:general_gradient_inexact_new}, \eqref{equ:function_error1}, and 
\eqref{equ:function_error2} hold. 
Thanks to Proposition~\ref{prop:error_control}, this can be done by 
properly controlling the termination parameter $\operatorname{res}$. 

Specifically, let  $\widetilde{b} \geq 0$, and $\{(\widetilde{\eta}_i^{\v g})^2\}$ and $\{\widetilde{\eta}_i^f\}$ be nonnegative summable sequences.  
By setting 
	\begin{equation}\label{equ:set_res}
		\operatorname{res} = \min\left\{\sqrt{(\widetilde{\eta}_i^{\v g})^2 + \widetilde{b}^2 \|\v x^i(\lambda) - \v x^i\|^2}, \widetilde{\eta}_i^f \right\},
	\end{equation}
then \eqref{equ:general_gradient_inexact_new}, \eqref{equ:function_error1}, and 
	\eqref{equ:function_error2} become 
	
	\begin{subequations}
		\begin{align*}
				&\|\varepsilon_i^{\v{g}}\|^2 \leq  (C_i \widetilde{\eta}_i^{\v{g}})^2 +  (C_i\widetilde{b})^2   \|\v{x}^i(\lambda) - \v{x}^i\|^2, 
				\label{equ:general_gradient_inexact_new:JBCP} 
				\\
				&|\varepsilon_i^f| \leq C_i \widetilde{\eta}_i^f,	
				\quad |\varepsilon_i^f(\lambda)| \leq \widetilde C_i \widetilde{\eta}_i^f, 
		\end{align*}
	\end{subequations}
	where $C_i = C(\v x^i)$ and $\widetilde C_i = C(\v x^i(\lambda))$. 
	Moreover, the LS condition \eqref{equ:B1} becomes 
	\begin{equation}
		f_i(\lambda) \leq f_i  - \theta \lambda^{-1}\| \v x^{i}(\lambda) - \v x^i \|^2 + \nu_i, 
		\label{equ:B1:JBCP}
	\end{equation}
	with $
		\nu_i =   \Upsilon_1(\lambda) {(C_i\widetilde{\eta}_i^{\v g})^2}  + (C_i + \widetilde C_i) \widetilde{\eta}_i^f. 
		$
	However, due to the unavailability of $C_i$ and $\widetilde C_i$, the following adjustments might be necessary to find a suitable stepsize such that the LS condition \eqref{equ:B1:JBCP} holds: given initial guesses for $C_i$ and $\widetilde C_i$, if the LS condition \eqref{equ:B1:JBCP} is not satisfied, we not only decrease the stepsize but also increase $C_i$ and $\widetilde C_i$. 
	Similar to Lemma \ref{lem:ls}, we can show that this procedure terminates after a finite number of steps.
	Furthermore, the termination criterion \eqref{equ:termination}  becomes
	\begin{equation}
		\| \v G_{\lambda_i}(\v x^i, \v g^i) \| + C_i \sqrt{(\widetilde{\eta}_i^{\v g})^2 + \widetilde{b}^2 \|\v x^{i+1} - \v x^i\|^2} \leq \epsilon,
		\label{equ:termination2}
	\end{equation}
	which, together with Lemma \ref{lem:KKT}, guarantees that the KKT violation is controlled by the stopping parameter $\epsilon$.

We now present the APIG-FP algorithm for solving 
problem  \eqref{equ:JBCP_PAPC_dual}, with details given in Algorithm~\ref{alg:APIG-FP}. 
Next, we establish the convergence result of the proposed algorithm APIG-FP based on the theoretical guarantee for the APIG framework. 
\begin{theorem}
	\label{thm:convergece_PIGA}
	Suppose that Assumption \ref{ass:lip} holds, and let $\{\v x^i\}$ be the sequence generated by Algorithm \ref{alg:APIG-FP} with $\epsilon > 0$. 
	If there  exists a constant $B_1 > 0$ such that $\|\v x^i\| \leq B_1$, then, Algorithm~\ref{frm:PIG} terminates within $N_s = \mathcal{O}(\epsilon^{-2})$ iterations, and $\v x^{N_s}$ is an $\epsilon$-stationary point of problem \eqref{equ:generic_problem}. 
\end{theorem}
\begin{proof}
	Since $\|\v x^i\| \leq B_1$, we do not need the level-set boundness in Assumption \ref{ass:2} in Theorem \ref{thm:1}. 
	To apply Theorem \ref{thm:1} and establish the desired result, it suffices to verify Assumption \ref{ass:1}, which is confirmed in  Appendix~\ref{apd:convergece_APIG-FP} in the Supplementary Material. 
\end{proof}

Theorem \ref{thm:convergece_PIGA} shows that the iteration complexity of APIG-FP for solving problem \eqref{equ:JBCP_PAPC_dual} is $\mathcal{O}(\epsilon^{-2})$. 
The parameters $(\widetilde{b}, \widetilde{{\eta}}_i^{\v g}, \widetilde{\eta}_i^f)$ in Algorithm \ref{alg:APIG-FP} controls the precision of solving the subproblem \eqref{equ:JBCP_PAPC_dual_sub} at each iteration. 
Compared to the PG method applied to problem \eqref{equ:JBCP_PAPC_dual}, APIG-FP has a lower per-iteration computational cost via solving the subproblem inexactly. 
As demonstrated in numerical experiments in the next subsection, such reduction in computational effort can significantly improve the overall efficiency of the APIG-FP algorithm.

\subsection{Simulation Results}\label{ss:3e}

In this subsection, we evaluate the performance of the proposed algorithm, APIG-FP, and compare it with existing state-of-the-art algorithms. 
We adopt the same system setup as in \cite{liu2021UplinkdownlinkDualityMultipleaccess}. 
Specifically, we consider a downlink cooperative cellular network with $M = 7$ single-antenna BSs serving $K = 7$ single-antenna users. 
The channels between the BSs and the users are generated using the Rayleigh fading model with zero mean and unit variance. 
In the following numerical experiments, the users' SINR targets $\{\overline{\gamma}_k\}$, the BSs' power limits $\{\overline{P}_m\}$, and the fronthaul capacities $\{\overline{C}_m\}$ are assumed to be identical across all users and BSs, denoted by $\overline{\gamma}$, $\overline{P}$, and $\overline{C}$, respectively. 
By default, we set $\overline{\gamma} = \overline{C} = 3$, which are typical parameter values as used in \cite{liu2021UplinkdownlinkDualityMultipleaccess} and \cite{park2013JointPrecodingMultivariate}. 
The power limit $\overline{P}$ is selected based on $\overline{\gamma}$ through trial and error to ensure that most randomly generated problem instances are feasible and at least one of the PAPCs is active at the optimal solution\footnote{If none of the PAPCs is active at the optimal solution, the generated feasible problem can be directly solved by the algorithm proposed in \cite{fan2022EfficientlyGloballySolving} without needing to update $\v x$, as $\v x = \v 0$ is the optimal Lagrange multiplier. }.

In the proposed APIG-FP algorithm, we set the parameters as follows: $\epsilon = 10^{-6}$,  $\theta = 10^{-4}$, $\alpha = 0.25$, $\lambda_{\min} = 10^{-10}$, $\lambda_{\max} = 10^{10}$, $\lambda_0^{(0)} = \lambda_1^{(1)} = 1$, $\varrho = 1.1$, $C_0 = \widetilde{C}_0 = 100$, with an initial point  $\v x^0 = \v 0$. 
The selection of parameters $(\widetilde{b}, \widetilde{\eta}_i^{\v g}, \widetilde{\eta}_i^f)$ 
is crucial in  APIG-FP. For $\widetilde{\eta}_i^f$, we set $\widetilde{\eta}_i^f = 10^{-\delta_1} (i+1)^{-\delta_2}$ for  $i \geq 0$, where $\delta_1 \geq 0$ and $\delta_2 > 1$. 
We consider two different settings for $\widetilde{b}$ and $\widetilde{\eta}_i^{\v g}$, leading to two specialized versions of  APIG-FP: 
\begin{itemize}
	\item \textbf{APIG-FP-A} (\textbf{a}bsolute gradient inexactness): $\widetilde{b}\equiv 0$ and $\widetilde{\eta}_i^{\v g}  = 10^{-\delta_1} (i+1)^{-\delta_2}$ for $i \geq 0$. 
	\item \textbf{APIG-FP-R} (\textbf{r}elative gradient inexactness):   $\widetilde{b} = 10^{-\delta_3}$ with $\delta_3 \geq 0$ and $\widetilde{\eta}_i^{\v g} \equiv 0$.  
\end{itemize}
All the results presented in this subsection are averaged over 200 randomly generated instances. 

\subsubsection{Behaviors of Proposed Algorithms}

\begin{table}[t]
	\setlength{\tabcolsep}{3pt}
	\caption{Average results of 200 randomly generated instances with $\overline{P} = 12$ versus different algorithm parameters.}
	\label{tab:2}
	\begin{center}
		\begin{tabular}{@{}ccccccccccccccccc@{}}
			\toprule
			& &\multicolumn{5}{c}{Time ($10^{-2}$ seconds)} & \multicolumn{5}{c}{Per-Iteration FPIs} & \multicolumn{5}{c}{Outer Iterations} \\
			\cmidrule(l){3-7} \cmidrule(l){8-12} \cmidrule(l){13-17}
			$\delta_1$ & $\delta_2$ & \multicolumn{1}{c}{PG} & \multicolumn{1}{c}{a} & \multicolumn{1}{c}{b} & \multicolumn{1}{c}{c} & \multicolumn{1}{c}{d} & \multicolumn{1}{c}{PG} & \multicolumn{1}{c}{a} & \multicolumn{1}{c}{b} & \multicolumn{1}{c}{c} & \multicolumn{1}{c}{d} & \multicolumn{1}{c}{PG} & \multicolumn{1}{c}{a} & \multicolumn{1}{c}{b} & \multicolumn{1}{c}{c} & \multicolumn{1}{c}{d} \\
			\midrule
			$1$ & $1.2$ & 7.8 & 2.4 & \textbf{2.4} & \textbf{2.5} & 2.8 & 182 & \textbf{7} & \textbf{6} & \textbf{7} & \textbf{8} & 10 & 63 & 56 & 51 & 48\\
			$1$ & $2.0$ & -- & 2.5 & 2.6 & 2.6 & 2.8 & -- & 9 & 8 & 9 & 10 & -- & 52 & 50 & 45 & 42\\
			$1$ & $4.0$ & -- & 3.2 & 3.1 & 3.1 & 2.8 & -- & 32 & 26 & 26 & 23 & -- & 24 & 23 & 23 & 23\\
			$2$ & $1.2$ & -- & \textbf{2.4} & 2.5 & 2.5 & 2.8 & -- & 9 & 8 & 9 & 11 & -- & 50 & 48 & 43 & 41\\
			$2$ & $2.0$ & -- & 2.4 & 2.6 & 2.6 & \textbf{2.7} & -- & 13 & 11 & 12 & 13 & -- & 38 & 39 & 37 & 34\\
			$2$ & $4.0$ & -- & 2.8 & 3.0 & 3.0 & 3.0 & -- & 40 & 35 & 36 & 35 & -- & \textbf{17} & \textbf{17} & \textbf{17} & \textbf{17}\\
			\bottomrule
		\end{tabular}
		\begin{minipage}{0.5\textwidth}
			\vspace{6pt}
			\footnotesize
			\textit{Note:} Here, ``a' stands for APIG-FP-A, ``b', ``c', and ``d' stand for APIG-FP-R with $\delta_3 = 0, 1, 2$, respectively. Besides, since PG does not depend on the parameters $(\delta_1, \delta_2)$, we only report its result in the first row of the PG column and use the dash mark ``--' for the remaining rows.
		\end{minipage}
	\end{center}	\vspace{-15pt}
\end{table}

We first study the performance of APIG-FP-A and APIG-FP-R by varying \(\delta_1\), \(\delta_2\), and \(\delta_3\). 
For comparison, the PG method \eqref{equ:PG} with the Armijo LS condition \eqref{equ:B1:JBCP} is used as a benchmark. 
While the standard PG method requires exact gradient and function evaluations (which are computationally prohibitive for problem \eqref{equ:JBCP_PAPC_dual} in practice), we implement it  through APIG-FP with high-precision settings: $\widetilde{b} = 0$ and $\widetilde{\eta}_i^{\v g} = \widetilde{\eta}_i^f = 10^{-10}$. 

The comparison results are reported in Table II. 
In this table, ``time'' refers to the average CPU time, ``Per-Iteration FPIs'' denotes the average number of FP iterations \eqref{equ:dual_iter} and \eqref{equ:primal_iter} per outer iteration, and ``Outer Iterations'' represents the average number of proximal gradient iterations.  
Additionally, the best performance in each column is marked in \textbf{bold} for both APIG-FP-A and APIG-FP-R. 

From Table~\ref{tab:2}, we make the following observations: 
(i) The PG method requires the fewest outer iterations but the most FP iterations per outer iteration. 
Overall, it is about $2.5$ to $3$ times slower than both APIG-FP-A and APIG-FP-R. 
(ii) Increasing \(\delta_1\), \(\delta_2\), or \(\delta_3\) (i.e., tightening the inexactness condition) generally reduces the number of outer iterations but increases the number of FP iterations per outer iteration. 
This demonstrates that the inexactness conditions require careful calibration to balance between inner fixed-point iterations and outer proximal gradient iterations. 
(iii) The CPU time performance of  APIG-FP-A and APIG-FP-R remains relatively stable as the parameters \(\delta_1\), \(\delta_2\), and \(\delta_3\) vary. 
Based on the results,  we use APIG-FP-A with $(\delta_1, \delta_2) = (2,1.2)$ and APIG-FP-R with $(\delta_1, \delta_2, \delta_3) = (1,1.2,1)$ in the following. 

\subsubsection{Comparison with Benchmarks}
To illustrate the efficiency of the proposed algorithms, we compare them with two state-of-the-art benchmarks: 
\begin{itemize}
	\item \textbf{SDR}: 
	This benchmark directly solves the SDR in \eqref{equ:JBCP_PAPC_SDR} using CVX \cite{CVX}. 
	For fair comparison, we set the precision of CVX to be the same $\epsilon = 10^{-6}$ as used in APIG-FP. 
	
	\item \textbf{PSG}: The projected subgradient (PSG) algorithm treats the gradient $\nabla d(\v x)$ as a subgradient. The diminishing stepsize is used and chosen as $\lambda_{i} = \lambda (i+1)^{-\delta}$ (to guarantee the convergence). 
	We systematically select and tune $\lambda \in \{10^{-2}, 10^{-1}, 1, 10\}$ and $\delta \in \{10^{-2}, 10^{-1}, 1, 10\}$ for the best performance of this benchmark. 
\end{itemize}

To verify the solution quality, we check that all test algorithms can return a solution $\v x$ satisfying  $\|\v G_{\lambda}(\v x, \nabla f(\v x))\| \leq 10^{-5}$ and  $|f(\v x) - f(\v x^\star)| \leq 10^{-6}$, with the function and gradient information computed via  the high-precision FP iterations \eqref{equ:dual_iter} and \eqref{equ:primal_iter}. 
Therefore, in Fig. \ref{fig:TvsdSINR}, we focus on computational efficiency and report the average CPU time and relative average CPU time versus the SINR targets for different algorithms. 
The relative average CPU time is defined as the runtime ratio between each algorithm and the fastest algorithm at each SINR target. 

From Fig. \ref{fig:TvsdSINR}, we make the following observations: (i) Both APIG-FP-A and APIG-FP-R significantly outperform PSG, achieving speedup from $4$ to $35$ times. This highlights the importance of  recognizing the differentiability of the objective function $d(\cdot)$ in problem \eqref{equ:JBCP_PAPC_dual}, as shown in Proposition \ref{prop:diff}.  
(ii) While the average CPU time of SDR remains stable across different SINR targets, our APIG-FP-A and APIG-FP-R are significantly faster, with speedups ranging from $15$ to $95$ times. This emphasizes the advantage  of leveraging the problem's structure in solving problem \eqref{equ:JBCP_PAPC_dual},  enabling us to develop computationally efficient first-order methods. 
(iii) Both APIG-FP-A and APIG-FP-R outperform PG by $2.5$ to $4.5$ times in CPU time, demonstrating the efficiency of the proposed APIG framework~in solving the JBCP with PAPCs. 
Moreover, the speedup increases as the SINR target grows. 

\begin{figure}[t]
	\centering
	\includegraphics[width=0.24\textwidth]{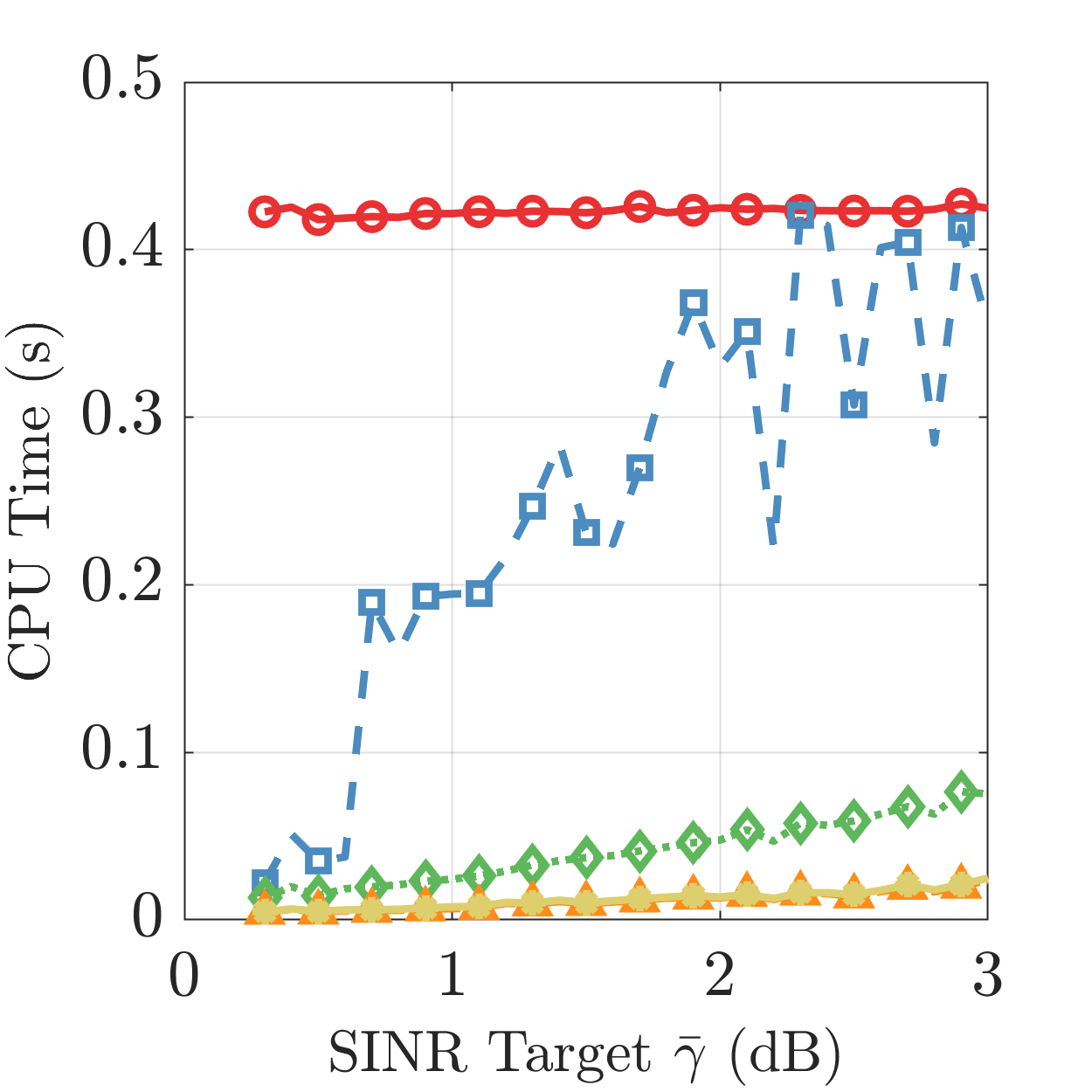}
	\includegraphics[width=0.24\textwidth]{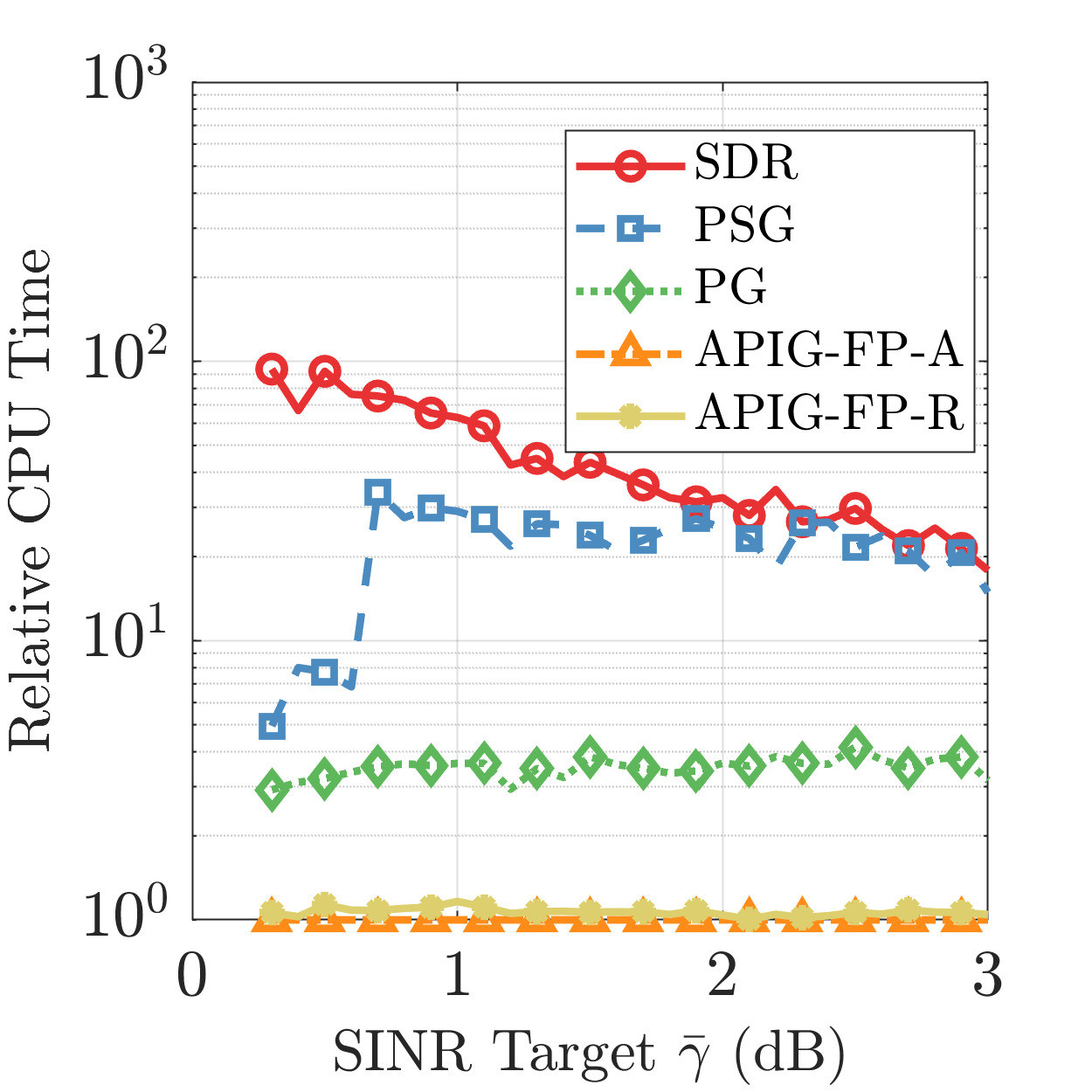}
	\vspace{-20pt}
	\caption{The average CPU time and relative average CPU time versus the SINR target $\overline{\gamma}$ for different algorithms. }
	\label{fig:TvsdSINR}	\vspace{-15pt}
\end{figure}

\section{Conclusion}
In this paper, we propose the APIG framework for solving a class of nonsmooth composite optimization problems, which have wide applications in signal processing and machine learning. Compared to existing methods, the proposed algorithmic framework accommodates errors in both function and gradient evaluations and is applicable a broad range of optimization problems. We show that the proposed framework is able to find an $\epsilon$-stationary point within $\mathcal{O}(\epsilon^{-2})$ iterations for nonconvex problems and an $\epsilon$-optimal solution within $\mathcal{O}(\epsilon^{-1})$ iterations for convex ones. By applying the APIG framework to the JBCP with PAPCs and judiciously exploiting the problem’s structure, we develop an efficient APIG-FP algorithm. Numerical simulations show that APIG-FP significantly outperforms state-of-the-art benchmark algorithms in computational efficiency. An interesting future work is to apply the APIG framework to other signal processing problems involving function and gradient errors. The key challenge is to effectively control these errors to achieve an optimal trade-off between the computational complexity of function and gradient approximations and the overall convergence of the framework.

\appendices
\section{Proof of Theorem \ref{thm:PIG_convergence}}\label{apd:PIG_convergence}
Our proof progressively considers the nonconvex case, the convex case I, and the convex case II. 
In the nonconvex case, we first show two key inequalities related to $\v x^i(\lambda)$. Then, we establish the $\mathcal{O}(\epsilon^{-2})$ iteration complexity by showing two key bounds in \eqref{equ:58} and \eqref{equ:bound:Delta:2}. 
In the convex cases I and II, we establish the desired conlusions bounding $F(\v x^{i+1}) - F^\star$ through the key inequalities in \eqref{equ:83} and \eqref{equ:Fi1:F*:case2}, respectively.

\subsection{Nonconvex Case}
First, we establish two key inequalities related to $\v x^i(\lambda)$, as shown in \eqref{equ:opti} and \eqref{equ:about_h}. 
By the definition of $\v x^i(\lambda)$ in \eqref{equ:x:lambda} and the property of the proximal mapping \eqref{equ:def_prox}, we have
\begin{equation}
	- \v g^i - \lambda^{-1} (\v x^{i}(\lambda) - \v x^i) \in \partial h(\v x^{i}(\lambda)), 
	\label{equ:proj}
\end{equation}
where $\partial h(\v x^{i}(\lambda))$ denotes the subdifferential of $h(\cdot)$ at $\v x^{i}(\lambda)$ \cite{beck2017FirstOrderMethodsOptimization}. 
In the following, we denote $\m G_\lambda^i = \m G_{\lambda}(\v x^i, \v g^i)$ for simplicity of notation. 
By the definition in \eqref{equ:def_gradient_mapping}, we have 
\begin{equation}
	\lambda \m G_\lambda^i = \v x^i - \v x^i(\lambda), 
	\label{equ:conversion}
\end{equation}
which, together with \eqref{equ:proj} and the convexity of $h(\cdot)$, implies 
\begin{equation}
	\begin{aligned}
		h(\v x^{i}(\lambda)) \leq 
		h(\v y) - \langle \v g^i - \m G_\lambda^i, \v x^{i}(\lambda) - \v y \rangle, ~\forall\,y\in \mathbb{R}^n.
	\end{aligned}
	\label{equ:opti}
\end{equation}
Setting $\v y = \v x^i$ in \eqref{equ:opti} and using \eqref{equ:conversion}, we have 
\begin{equation}
	h(\v x^{i}(\lambda)) \leq h(\v x^i) + \lambda \langle \v g^i - \v G_\lambda^i, \v G_\lambda^i \rangle. 
	\label{equ:about_h}
\end{equation}

By \eqref{equ:def_etai}, \eqref{equ:def_etaif}, \eqref{equ:def_nu}, \eqref{equ:cond_C}, and \eqref{equ:lambda_bd}, we have
\begin{equation}
	\sum_{i=0}^{\infty}(\nu_i + 2 \eta_i^f) < \infty. 
	\label{equ:summable_2}
\end{equation}
Summing \eqref{equ:descent_property} for $i=0,1,\d,N-1$ with $N \geq 1$ and using the lower bound of $\lambda_i$ as shown in \eqref{equ:lambda_bd}, we have 
\begin{equation}
	\begin{aligned}
		\sum_{i=0}^{N-1} \| \v G_{\lambda_i}^i \|^2  \leq \frac{2 \Upsilon_2}{\theta \underline{\lambda}_{\theta}},
	\end{aligned}\label{equ:58} 
\end{equation}
where  $\Upsilon_2 =F(\v x^0) - F^\star + \sum_{i=0}^{\infty} ( \nu_i + 2 \eta_i^f )$. Here, $F^\star$ is the optimal value of problem \eqref{equ:generic_problem} and $\Upsilon_2 < \infty$ due to  \eqref{equ:summable_2}. 
From the definition of $\Delta_i^{\v g}$ in \eqref{equ:termination}  and \eqref{equ:conversion}, we have 
\begin{equation}
	\begin{aligned}
		\left(\Delta_i^{\v g}\right)^2 \leq{}& 2 (1 + a^2 + b^2 \lambda_i^2) \| \v G_{\lambda_i}^i \|^2 + 2 (\eta_i^{\v g})^2  \\
		\leq{}&  2 (1 + a^2 + b^2 \lambda_{\max}^2) \| \v G_{\lambda_i}^i \|^2 + 2 (\eta_i^{\v g})^2,
	\end{aligned}
	\label{equ:bound:Delta:2}
\end{equation}
where the second inequality uses the upper bound of $\lambda_i$ in \eqref{equ:lambda_bd}.   Combining \eqref{equ:58} and \eqref{equ:bound:Delta:2}, and using \eqref{equ:def_etai}, we have
\begin{equation}
	\label{equ:stop}
	\sum_{i = 0}^{N-1} \left(\Delta_i^{\v g}\right)^2  \leq \Upsilon_3 < \infty,
\end{equation}
where $\Upsilon_3 = \frac{4 (1 + a^2 + b^2 \lambda_{\max}^2)\Upsilon_2}{\theta \underline{\lambda}_{\theta}} + 2 \sum_{i = 0}^\infty (\eta_i^{\v g})^2$.  
Hence, $\{\Delta_i^{\v g}\}$ converges to zero, which further implies that the algorithm must terminate at some $N_s$ iterations, and for $i = 0,1, \ldots, N_s-1$, there holds that $\Delta_i^{\v g} > \epsilon$. 
Combining this and \eqref{equ:stop}, we have the desired  
$N_s <  \Upsilon_3 \epsilon^{-2} = \mathcal{O}(\epsilon^{-2})$.  

\subsection{Convex Case I}
When $f(\cdot)$ is convex, it follows directly from the controlled descent property \eqref{equ:descent_property} that 
\begin{equation} \label{equ:F:level}
	F(\v x^{i+1}) \leq F(\v x^{0}) + \sum_{i=0}^\infty (\nu_i + 2\eta_i^f) < \infty. 
\end{equation}
Let $\v x^\star \in \mathcal X^\star$ be an optimal solution to problem \eqref{equ:generic_problem}. 
By \eqref{equ:F:level} and Assumption \ref{ass:2}, there exists a constant $B_2 > 0$ such that 
\begin{equation}
	\| \v x^i - \v x^\star \| \leq B_2, \quad \forall\,i \geq 0. 
	\label{equ:bound_iterate}
\end{equation}
Define 
\begin{equation}\label{equ:def:Delta:i:F}
	\Delta_i^F = F(\v x^i) - F(\v x^\star). 
\end{equation}

For convex case I, the LS condition  \eqref{equ:B2} is used and $a = b = 0$, $\sum_{i}^\infty \eta_i^{\v g} < \infty$. Therefore,  by \eqref{equ:general_gradient_inexact_new}, we have 
\begin{equation}\label{equ:sum:eta:g}
	\sum_{i=0}^\infty \|\varepsilon_i^{\v g}\| \leq \sum_{i = 0}^\infty \eta_i^{\v g} <  \infty. 
\end{equation} 
Moreover, substituting \eqref{equ:conversion} into \eqref{equ:B2}, we have 
\begin{equation}
	f_i(\lambda_i) \leq f_i + \left\langle\v g^i, \v x^{i+1} - \v x^i \right\rangle + \frac{\lambda_i}{2} \|\v G_{\lambda_i}^i\|^2 + \nu_i.
	\label{equ:B2:new}
\end{equation}
By  \eqref{equ:function_error1}, \eqref{equ:function_error2}, and \eqref{equ:conversion}, we further have 
\begin{equation}
	\label{equ:f:i1:bound}
	f(\v x^{i+1}) \leq f_i(\lambda_i) + \eta_i^f + c \lambda_i\|\v G_{\lambda_i}^i\|^2
\end{equation}
and 
\begin{equation}
	\label{equ:f:i:bound}
	f_i \leq f(\v x^i) + \eta_i^f + c \lambda_i\|\v G_{\lambda_i}^i\|^2. 
\end{equation}
Combining \eqref{equ:B2:new}, \eqref{equ:f:i1:bound}, and \eqref{equ:f:i:bound}, and using the upper bound of $\lambda_i$ in \eqref{equ:lambda_bd}, we get 
\begin{equation}
	f(\v x^{i+1}) \leq f(\v x^i) + \left\langle\v g^i, \v x^{i+1} - \v x^i \right\rangle  + \frac{\lambda_i}{2} \|\v G_{\lambda_i}^i\|^2  + \widetilde \nu_i,
	\label{equ:B2:new2}
\end{equation}
where 
\begin{equation}\label{equ:tilde:nui}
	\widetilde \nu_i = 2 c \lambda_{\max} \|\v G_{\lambda_i}^i\|^2 + \nu_i + 2\eta_i^f.
\end{equation}
By the convexity of $f(\cdot)$, we have 
\begin{equation}
	f(\v x^{i}) \leq f(\v x^\star) + \langle \nabla f(\v x^i), \v x^i - \v x^\star\rangle. \nonumber %\label{equ:52}
\end{equation}
This, together with \eqref{equ:B2:new2}, \eqref{equ:def_eg},  and  \eqref{equ:conversion}, implies 
\begin{equation} \label{equ:fi1:f*}
	\begin{aligned}
		f(\v x^{i+1}) \leq{}&  f(\v x^\star)  + \left\langle\v g^i, \v x^{i+1} - \v x^\star \right\rangle - \langle \varepsilon_i^{\mathbf{g}}, \v x^i - \v x^\star\rangle \\ 
		{}&  + \frac{1}{2 \lambda_i} \|\v x^{i+1} - \v x^i\|^2 + \widetilde \nu_i. 
	\end{aligned}
\end{equation}

Setting $\lambda = \lambda_i$ and $\v y = \v x^\star$ in \eqref{equ:opti}, and using \eqref{equ:conversion}, we have 
\begin{equation}
	\begin{aligned}
		h(\v x^{i+1}) \leq{}& h(\v x^\star) - \langle \v g^i , \v x^{i+1}  - \v x^\star \rangle  \\ 
		{}& + \frac{1}{\lambda_i}\langle \v x^i - \v x^{i+1}, \v x^{i+1}  - \v x^\star \rangle.
	\end{aligned}
	\label{equ:opti:2}
\end{equation}
Using $F(\cdot) = f(\cdot) + h(\cdot)$ and combining \eqref{equ:bound_iterate}, \eqref{equ:fi1:f*}, and \eqref{equ:opti:2} along with the fact that $\|\v a\|^2  - 2 \langle \v a,  \v b \rangle  = \| \v b- \v a\|^2 - \|\v b\|^2$ 
with $\v a = \v x^{i+1} - \v x^i$ and $\v b = \v x^{i+1}  - \v x^\star$, 
we obtain 
\[
\Delta_{i+1}^F \leq B_2 \|\varepsilon_i^{\mathbf{g}}\|  + \widetilde \nu_i  +  \frac{1}{2\lambda_i} \left( \|\v x^i - \v x^*\|^2  - \|\v x^{i+1} - \v x^*\|^2\right), 
\]
which, together with \eqref{equ:lambda_bd}, further implies
\begin{multline}
	\Delta_{i+1}^F \leq \frac{\lambda_{\max}}{\underline{\lambda}} \left(B_2 \|\varepsilon_i^{\mathbf{g}}\|+ \widetilde \nu_i\right)  \\
	+ \frac{1}{2\underline{\lambda}} \left( \|\v x^i - \v x^*\|^2  - \|\v x^{i+1} - \v x^*\|^2\right). 
	\label{equ:83}
\end{multline}
Since the algorithm stops at the $N_s$-th iteration, summing \eqref{equ:83} over $i = 0, 1,\ldots, N - 1$ with $1 \leq N \leq N_s$ and adding $\Delta_0^F$ to both sides of the derived inequality, we have 
\[
\sum_{i = 0}^{N}  \Delta_{i}^F \leq \Upsilon_4, \quad \forall\, N \leq N_s, 
\]
where 
\[\Upsilon_4 = \Delta_0^F + \frac{1}{2\underline{\lambda}} \|\v x^0 - \v x^\star\|^2 + \frac{\lambda_{\max}}{\underline{\lambda}} \sum_{i = 0}^\infty \left(B_2 \|\varepsilon_i^{\mathbf{g}}\|  + \widetilde \nu_i \right). 
\]
From \eqref{equ:summable_2},  \eqref{equ:58}, \eqref{equ:sum:eta:g}, and \eqref{equ:tilde:nui},  we know that $\Upsilon_4 < \infty$. 
By the convexity of $F(\cdot)$ and \eqref{equ:def:Delta:i:F},  the averaged iterate $\overline{\v x}^{N}$ satisfies 
\[
F( \overline{\v x}^{N}) - F^\star \leq  \frac{1}{N + 1}\sum_{i = 0}^{N}\Delta_{i}^F \leq \frac{\Upsilon_4}{N+ 1},
\]
which is the desired result. 

\subsection{Convex Case II}
In this case, $\eta_i^{\v g} \equiv \eta_i^f \equiv 0$, and either the LS condition \eqref{equ:B1} or \eqref{equ:B2} is used. 
As a result, the controlled descent property becomes a sufficient descent property: 
\begin{equation}
	F(\v x^{i+1}) - F(\v x^i) \leq - \frac{\theta \lambda_i}{2}  \|\v G_{\lambda_i}^i\|^2. 
	\label{equ:descent_property:new}
\end{equation}
Additionally, the condition \eqref{equ:general_gradient_inexact_new} becomes 
\begin{equation}
	\|\varepsilon_i^{\v{g}}\| \leq  \sqrt{a^2 + b^2 \lambda_i^{2}} \|\v G_{\lambda_i}^i\| \leq (a + b\lambda_i) \|\v G_{\lambda_i}^i\|,
	\label{equ:general_gradient_inexact_new:new}
\end{equation}
where we use \eqref{equ:conversion}. 
From the convexity of $f(\cdot)$, it follows that 
\begin{equation}
	f(\v x^{i+1}) \leq f(\v x^\star) + \langle \nabla f(\v x^{i+1}), \v x^{i+1} - \v x^\star\rangle, \label{equ:78}
\end{equation}
Setting $\lambda = \lambda_i$ and $\v y = \v x^\star$ in \eqref{equ:opti}, and adding it to \eqref{equ:78} gives 
\begin{equation}\label{equ:Fi1:F*:case2}
	\begin{aligned}
		&F(\v x^{i+1})  - F(\v x^\star) \\
		\leq{}& \langle \nabla f(\v x^{i+1}) - \v g^i + \v G_{\lambda_i}^i, \v x^{i+1} - \v x^\star \rangle \\
		\overset{\text{(a)}}{\leq}{}& B_2 \left( \|\nabla f(\v x^{i+1}) - \nabla f(\v x^{i})\| + \|\nabla f(\v x^{i})  - \v g^i\| + \|\v G_{\lambda_i}^i \| \right) \\
		\overset{\text{(b)}}{\leq}{}& B_2 \left(L \lambda_i \|\v G_{\lambda_i}^i\| + \| \varepsilon_i^{\v g} \| + \|\v G_{\lambda_i}^i\|\right) \\
		\overset{\text{(c)}}{\leq}{}& \Upsilon_4 \|\v G_{\lambda_i}^i\|,
	\end{aligned}
\end{equation}
where $\Upsilon_4 = B_2( (L + b) \lambda_{\max}  + a +  1)$. 
In the above, (a) follows from \eqref{equ:bound_iterate},  (b) is due to the $L$-Lipschitz continuity of $\nabla f(\cdot)$ as in Assumption \ref{ass:1}, the definition of $\varepsilon_i^{\v g}$ in \eqref{equ:def_eg}, and \eqref{equ:conversion}, and (c) uses the upper bound of $\lambda_i$ in \eqref{equ:lambda_bd} and \eqref{equ:general_gradient_inexact_new:new}.  

Substituting \eqref{equ:Fi1:F*:case2} into  \eqref{equ:descent_property:new}, and using the definition \eqref{equ:def:Delta:i:F} along with the lower bound on $\lambda_i$ in \eqref{equ:lambda_bd}, we have 
\[
\Delta_{i}^F - \Delta_{i+1}^F \leq  \frac{\theta \underline{\lambda}_{\theta} }{2 \Upsilon_4^2}(\Delta_{i+1}^F)^2.
\]
Applying \cite[Lemma 11.17]{beck2017FirstOrderMethodsOptimization} to the sequence $\{\Delta_i^F\}$, we know that for any $2 \leq N \leq N_s$, there holds that 
\[
\Delta_N^F  \leq  \max \left\{\left(\frac{1}{\sqrt{2}}\right)^{N-1} \Delta_0^F,  \frac{1}{N-1} \cdot \frac{8\Upsilon_4^2}{\theta \underline{\lambda}_{\theta}} \right\}, 
\]
which completes the proof.

\IEEEpeerreviewmaketitle

\bibliographystyle{IEEEtran}
\bibliography{papc_tsp}

\clearpage
\begin{center}
  \Large\textbf{Supplementary Material}
\end{center}
\appendices
\setcounter{section}{2}
\setcounter{page}{1}
\vspace{-20pt}
\section{Proof of Lemma \ref{lem:ls} and Lemma \ref{lem:descent_property}}
\subsection{Proof of Lemma \ref{lem:ls}}
\label{apd:ls} 
From the $L$-Lipschitz continuity of $\nabla f(\cdot)$ in Assumption \ref{ass:1},  \eqref{equ:conversion}, and the definition of $\varepsilon_i^{\v g}$ in \eqref{equ:def_eg}, we have
\begin{equation}
	\begin{aligned}
		f(\v x^i(\lambda)) \leq f(\v x^i) + \lambda \langle \varepsilon_i^{\v g}, \v G_\lambda^i \rangle - \lambda \langle \v g^i, \v G_\lambda^i \rangle  +  \frac{L\lambda^2}{2} \| \v G_\lambda^i \|^2. 
	\end{aligned}
	\label{equ:appedix:f:x:lambda:upper}
\end{equation}
Let 
\begin{equation}\label{equ:xi}
	\xi=  
	\begin{cases}
		1,& \text{if $a^2 + b^2 = 0$}, \\
		\frac{\lambda}{\sqrt{a^2 + b^2 \lambda^2}}, & \text{otherwise}.
	\end{cases}
\end{equation}
By the Cauchy-Schwartz inequality, we have 
\[
\begin{aligned}
	\lambda \langle \varepsilon_i^{\v g}, \v G_{\lambda}^i \rangle \leq{}&   \frac{\xi}{2} \| \varepsilon_i^{\v g} \|^2 + \frac{\lambda^2}{2\xi} \| \v G_\lambda^i \|^2  \\
	\leq{}&  \frac{\xi}{2} (\eta_i^{\v g})^2 + \frac{1}{2} \left( \xi (a^2 + b^2 \lambda^2) + \frac{\lambda^2}{\xi}  \right) \| \v G_\lambda^i \|^2,
\end{aligned}
\]
where the second inequality uses  \eqref{equ:general_gradient_inexact_new} and \eqref{equ:conversion}.  Substituting the above into 
\eqref{equ:appedix:f:x:lambda:upper} yields
\begin{equation}
	\begin{aligned}
		{}&f(\v x^i(\lambda))\\
		\leq{} &  f(\v x^i) + \frac{\xi}{2} (\eta_i^{\v g})^2 + \frac{1}{2} \left( \xi (a^2 + b^2 \lambda^2) + \frac{\lambda^2}{\xi}  \right) \| \v G_\lambda^i \|^2 \\
		& - \lambda \langle \v g^i, \v G_\lambda^i \rangle +  \frac{L\lambda^2}{2} \| \v G_\lambda^i \|^2. 
	\end{aligned}
	\label{equ:37}
\end{equation}
Using the definition of $\Upsilon_1(\lambda)$ in \eqref{equ:cond_C} and the fact that $\max\{a, b\lambda\} \leq \sqrt{a^2 + b^2 \lambda^2} \leq a + b \lambda$, we have, from \eqref{equ:xi} and \eqref{equ:37}, that 
\begin{equation}
	\begin{aligned}
		f(\v x^i(\lambda)) \leq{}& f(\v x^i) + \frac{\lambda}{2}\big(2a + (2b + 1 + L) \lambda\big) \| \v G_\lambda^i\|^2 \\
		{}& +\Upsilon_1(\lambda) (\eta_i^{\v g})^2 - \lambda \langle \v g^i, \v G_\lambda^i \rangle. 
	\end{aligned}
	\label{equ:59}
\end{equation}
By \eqref{equ:function_error1}, \eqref{equ:function_error2}, and \eqref{equ:conversion}, we have $f(\v x^i(\lambda)) \geq f_i(\lambda) - \eta_i^f - c\lambda\|\v G_{\lambda}^i\|^2$ 
and $f(\v x^i) \leq f_i +  \eta_i^f + c\lambda\|\v G_{\lambda}^i\|^2$. 
Combining these with \eqref{equ:59} and the definition of $\nu_i$ in \eqref{equ:def_nu}, we obtain  
\begin{equation}
	\begin{aligned}
		f_i(\lambda) 
		\leq{}& f_i   -\lambda \langle \v g^i, \v G_\lambda^i \rangle   \\
		{}& + \frac{\lambda}{2} \left(2a + 4c + (2b  + 1 + L)\lambda\right) \| \v G_\lambda^i\|^2 + \nu_i.
	\end{aligned}
	\label{equ:interm}
\end{equation}
Adding \eqref{equ:about_h} and \eqref{equ:interm} gives 
\begin{equation}
	\begin{aligned}
		&f_i(\lambda) + h(\v x^i(\lambda)) \\
		\leq{} &f_i + h(\v x^i) \\
		&+ \frac{\lambda}{2} \left(2a+ 4c - 2+(2b+1+L)\lambda\right)  \| \v G_\lambda^i \|^2 + \nu_i. 
	\end{aligned}
	\label{equ:interm2}
\end{equation}
By leveraging  the relation in \eqref{equ:conversion}, the definition of $\overline \lambda_{\theta}$ in \eqref{equ:def_tildelambda:theta}, and the bounding inequalities from \eqref{equ:interm} and \eqref{equ:interm2}, we conclude that the LS condition \eqref{equ:B1} is satisfied for all $\lambda \in (0, \overline{\lambda}_{\theta}]$, and the LS condition
\eqref{equ:B2} holds for all $\lambda \in (0, \overline{\lambda}_{1/2}]$.

\subsection{Proof of Lemma \ref{lem:descent_property}}
\label{apd:finite_termination}

First, using \eqref{equ:conversion},  the inequality \eqref{equ:about_h} becomes 
\begin{equation*}
	h(\v x^{i}(\lambda)) \leq h(\v x^i) +  \langle \v g^i, \v x^i - \v x^i(\lambda) \rangle - \frac{1}{\lambda} \|\v x^i(\lambda) - \v x^i\|^2. 
	\label{equ:about_h:new}
\end{equation*}
Adding this to the LS condition \eqref{equ:B2} implies that the LS condition \eqref{equ:B1} holds with $\theta = 1/2$. 
Therefore, it suffices to show the desired property for the LS condition \eqref{equ:B1} with $\theta \in (0, 1)$. 
If \eqref{equ:B1} holds, then setting $\lambda = \lambda_i$ and using \eqref{equ:conversion}, we have 
\begin{equation}
	f_i(\lambda_i) + h(\v x^{i+1}) \leq f_i + h(\v x^i) - \theta \lambda_i \| \v G_{\lambda_i}^i \|^2 + \nu_i. 
	\label{equ:61}
\end{equation}
Using \eqref{equ:def_eif} and \eqref{equ:def_eif_lambda} with $\lambda = \lambda_i$, we obtain $f_i(\lambda_i) = f(\v x^{i+1}) + \varepsilon_i^f(\lambda_i)$  and $f_i = f(\v x^i) +  \varepsilon_i^f(\lambda_i)$. By substituting these two expressions into \eqref{equ:61}, leveraging the relationship $F(\cdot)=f(\cdot)+h(\cdot)$ from problem \eqref{equ:generic_problem}, and combining the bounds in \eqref{equ:function_error1} and \eqref{equ:function_error2} under the condition $c \leq \theta/4$ (after \eqref{equ:function_error2}), we directly derive the desired inequality \eqref{equ:descent_property}. 

\section{Useful Properties of $I_{\v x}(\cdot)$, $J_{\vv \beta, \v x}(\cdot)$, and $\mu(\cdot, \cdot)$}\label{apd:lemmas}
In this section, we first give Lemma \ref{lem:fixedpoint} and Lemma \ref{lem:properties}, which describe the properties of mappings $I_{\v x}(\cdot)$ and $J_{\vv \beta, \v x}(\cdot)$, and then show Lemma \ref{lem:norm_control} about Thompson’s metric $\mu(\cdot, \cdot)$. 
\begin{lemma}
	For any given $\v x \geq 0$, the mapping $I_{\v x}(\cdot)$ has a unique FP, denoted as $\vv \beta^\star(\v x)$, and the mapping $J_{\vv \beta^\star(\v x), \v x}(\cdot)$ has a unique solution, denoted as $\vv p^\star(\v x)$. 
	Moreover, for $\widetilde{\vv \beta} \geq \v 0$ such that $J_{\widetilde{\vv \beta}, \v x}(\cdot)$ has a FP, then the mapping $J_{\widetilde{\vv \beta}, \v x}(\cdot)$ has a unique FP, denoted as  $\vv p^\star(\widetilde{\vv \beta}, \v x)$. 
	\label{lem:fixedpoint}
\end{lemma}
\begin{proof}
	For any $\v x \geq \v 0$, by \cite[Lemma 1]{fan2022EfficientlyGloballySolving}, the mappings $I_{\v x}(\cdot)$, $J_{\widetilde{\vv \beta}, \v x}(\cdot)$, and $J_{\vv \beta^\star(\v x), \v x}(\cdot)$ each has at most one FP. 
	The strict feasibility of problem \eqref{equ:JBCP_PAPC_SDR} implies the strict feasibility of problem \eqref{equ:JBCP_PAPC_dual_sub} for any $\v x \in \mathbb{R}_+^M$, ensuring that both $I_{\v x}(\cdot)$ and $J_{\vv \beta^\star(\v x), \v x}(\cdot)$ have exactly one FP. 
	Since $\widetilde{\vv \beta}$ is selected such that $J_{\widetilde{\vv \beta}, \v x}(\cdot)$ has a FP, this mapping has exactly one FP. 
\end{proof}

Given  $B > 0$ and an positive integer $n$, define 
\begin{equation*}
	\mathcal S_B^n = \{ \v z \in \mathbb{R}^n_+ \mid \|\v z\| \leq B\}. 
	\label{equ:def_S}
\end{equation*}
We now present the following lemma. 

\begin{lemma}
	\label{lem:properties}
	The following properties hold for mappings $I_{\v x}(\cdot)$ and $J_{\vv \beta, \v x}(\cdot)$ defined in \eqref{equ:def_I} and \eqref{equ:def_Jk}, respectively. 
	\begin{itemize}
		\item[(a)] For any $\v x\geq \v 0$, $I_{\v x}(\vv \beta)$ is a rational function, i.e., its $k$-th component can be written as 
		$$
		(I_{\v x}(\vv \beta))_k  = \frac{P_k(\vv \beta, \v x)}{Q_k(\vv \beta, \v x)}, 
		$$
		where $P_k(\vv \beta, \v x)$ and $Q_k(\vv \beta, \v x)$ are polynomials with respect to both $\vv \beta$ and $\v x$. 
		\item[(b)] 
		For any given $\v x \geq \v 0$, and some $B(\v x) > 0$, there exists $\kappa_2(\v x) \in (0, 1)$ such that
		\begin{equation}
			\mu(I_{\v x}(\vv \beta), I_{\v x}(\overline{\vv \beta})) \leq \kappa_2(\v x) \mu(\vv \beta, \overline{\vv \beta}) 
			\label{equ:lc1}
		\end{equation}
		for any $\vv \beta, \overline{\vv \beta} \in \mathcal S_{B(\v x)}^K$. 
		For any $B > 0$, there exists $\kappa_4(B) \in (0, 1)$ such that
		\begin{equation}
			\mu(I_{\v x}(\vv \beta), I_{\v x}(\overline{\vv \beta})) \leq \kappa_4(B) \mu(\vv \beta, \overline{\vv \beta})
			\label{equ:lc2}
		\end{equation}
		for any $\vv \beta,\overline{\vv \beta} \in \mathcal S_{B}^K$ and $ \v x \in \mathcal S_{B}^M$. 
		\item[(c)] For any given $\v x \geq \v 0$ and some $B(\v x) > 0$, there exists $\kappa_3(\v x) \in (0, 1)$ such that
		\begin{equation}
			\mu(J_{\vv \beta^\star(\v x),\v x}(\v p), J_{\vv \beta^\star(\v x),\v x}(\overline{\v p})) \leq \kappa_3(\v x) \mu(\v p, \overline{\v p})
			\label{equ:lc3}
		\end{equation}
		for any $\v p, \overline{\v p} \in \mathcal S_{B(\v x)}^K$. 
	\end{itemize}
\end{lemma}

\begin{proof}
	First, we prove (a). For any matrix $\m C \succ \m 0$ and $\v h \in \mathbb{R}^{M}$, the Cramer's rule shows that $\m C^{-1} \v h$ can be represented componentwise by the ratio of the determinants of two matrices using entries of $\m C$ and $\v h$. 
	Hence, $\m C^{-1} \v h$ is a rational function with respect to the entries of $\m C$. 
	Combining this with the fact that $\m \Lambda_{\v x, m}(\cdot)$ is rational from its definition \cite{fan2023QoSbasedBeamformingCompression}, and the composition rule of rational functions, $I_{\v x}(\cdot)$ defined in \eqref{equ:def_I} is rational. 
	
	Then, we prove (b). 
	For a matrix $\m A$, let $\rho(\m A)$ denote the spectral radius of $\m A$. 
	It follows from Eq. (45) in \cite{fan2023QoSbasedBeamformingCompression} that for any $B > 0$, $\vv \beta, \overline{\vv \beta} \geq \v 0$, and $\v x \geq \v 0$, 
	\begin{equation}
		\begin{aligned}
			&\frac{\mu(I_{\v x}(\vv \beta), I_{\v x}(\overline{\vv \beta}))}{\mu(\vv \beta, \overline{\vv \beta})} \\
			\leq{} &\max_k \Big\{\varphi (\mathrm{e}^{\mu(\vv \beta, \overline{\vv \beta})}, \lambda_1^k(\vv \beta, \v x)), \varphi (\mathrm{e}^{\mu(\vv \beta, \overline{\vv \beta})}, \lambda_1^k(\overline{\vv \beta}, \v x))\Big\} \\
			:={} &\overline \varphi(\vv \beta, \overline{\vv \beta}, \v x) \in [0,1),
		\end{aligned}
		\label{equ:dual_rate}
	\end{equation}
	where $\lambda_1^k(\vv \beta, \v x) = \rho(\m C(\vv \beta, \{\m \Lambda_{\v x, m}(\vv \beta)\}, \v x) - \m I)$ and $\varphi(\alpha, \lambda) = \log_\alpha \left(\frac{1+\alpha \lambda}{1 + \lambda}\right)$. 
	Due to the continuity of $\overline \varphi$ with respect to $(\vv \beta, \overline{\vv\beta}, \v x)$, we can set $\kappa_2(\v x) = \max_{\vv \beta, \overline{\vv \beta}\in \mathcal{S}_{B(\v x)}^{K}} \overline \varphi(\vv \beta, \overline{\vv \beta}, \v x)$ as well as  $\kappa_4 = \max_{\v x \in \mathcal{S}_{B}^{M}} \max_{\vv \beta, \overline{\vv \beta}\in \mathcal{S}_{B}^{K}} \overline \varphi(\vv \beta, \overline{\vv \beta}, \v x)$, and rearrange \eqref{equ:dual_rate} to obtain \eqref{equ:lc1} and \eqref{equ:lc2}.	
	
	Finally, we prove (c). 
	We begin by simplifying the expression for $J_{\vv \beta^\star(\v x), \v x}(\cdot)$. 
	Plugging \eqref{equ:def_u} and \eqref{equ:def_Q} into \eqref{equ:def_Jk} gives 
	\begin{equation*}
		\begin{aligned}
			(J_{\vv \beta^\star(\v x), \v x}(\v p))_k =  \frac{\gamma_k (\v a_k(\v x)^\transpose \v p + \sigma_k^2)}{|\v h_k^\hermitian \v u_k(\vv \beta^\star(\v x), \v x)|^2},  
		\end{aligned}
		\label{equ:def_J}
	\end{equation*}
	where  $\left(\v a_k(\v x)\right)_j$ is $\v h_k^\hermitian \underline{\m A}_{j}(\vv \beta^\star(\v x), \v x) \v h_k$ if $j = k$ and is $\v h_k^\hermitian \underline{\m A}_{j}(\vv \beta^\star(\v x), \v x) \v h_k + |\v h_k^\hermitian \v u_{j}(\vv \beta^\star(\v x), \v x)|^2$ otherwise.
	Next, we follow a similar approach as in \cite{fan2023QoSbasedBeamformingCompression} to derive $\kappa_3(\v x)$. 
	%	We temporarily drop the dependence on $\v x$ in $\v a_k$ and $r_k$ for brevity. 
	For any  $\alpha > 1$, $\v p \geq \v 0$, $\v x \geq \v 0$, and $k \in \cK$, we have,
	\begin{equation}
		\begin{aligned}
			\frac{(J_{\vv \beta^\star(\v x), \v x}(\alpha \v p))_k}{(J_{\vv \beta^\star(\v x), \v x}(\v p))_k}
%			&=\frac{\v a_k^\transpose(\v x) \v p \alpha + \sigma_k^2}{\v a_k^\transpose(\v x) \v p + \sigma_k^2} \\
			&= \frac{\v a_k^\transpose(\v x) \v p}{\v a_k^\transpose(\v x) \v p + \sigma_k^2} (\alpha - 1) + 1 \\
			&\overset{\text{(a)}}{\leq} (1 + \alpha - 1)^{\frac{\v a_k^\transpose(\v x) \v p}{\v a_k^\transpose(\v x) \v p + \sigma_k^2}} = \alpha^{\frac{\v a_k^\transpose(\v x) \v p}{\v a_k^\transpose(\v x) \v p + \sigma_k^2}}, 
		\end{aligned}\label{equ:102a}
	\end{equation}
	where (a) follows from Bernoulli's inequality. 
	Define  $\lambda_2(\v p, \v x) = \max_k \left\{ \frac{\v a_k^\transpose(\v x) \v p}{\v a_k^\transpose(\v x) \v p + \sigma_k^2} \right\} \in [0, 1)$. 
	Notice that $\lambda_2(\v p, \v x)$ is a continuous function with respect to $\v p$. Define $\kappa_3(\v x)= \max_{\v p \in \mathcal{S}_{B(\v x)}^{K}} \lambda_2(\v p, \v x) $. 
	Then, applying \cite[Lemma 1]{fan2023QoSbasedBeamformingCompression} to \eqref{equ:102a} gives
	\[
	\begin{aligned}
		\frac{\mu(J_{\vv \beta^\star (\v x),\v x}(\v p), J_{\vv \beta^\star,\v x}(\overline {\v p}))}{\mu(\v p, \overline{\v p})} 
		\leq{}& \max \{ \lambda_2(\v p, \v x), \lambda_2(\overline {\v p}, \v x)\} \\
		\leq{}& \kappa_3(\v x), 
	\end{aligned}
	\]
	which implies \eqref{equ:lc3} by rearranging.

\end{proof}

\begin{lemma}\label{lem:norm_control}
	For  any  $\v x, \v y \in \mathbb{R}^K_+ \setminus \{\v 0\}$, we have
	\begin{equation*}
		\|\v x - \v y\| \leq  \max\{\|\v x\|, \|\v y\|\} K \mu(\v x, \v y). 
	\end{equation*}
	%	where  $\mathcal S_B^M$ is defined in \eqref{equ:def_S}. 
\end{lemma}
\begin{proof} 
	Let $x_k$ and $y_k$ represent the $k$-th components of $\v x$ and $\v y$, respectively. For  any  $\v x, \v y \in \mathbb{R}^K_+ \setminus \{\v 0\}$, we have $0 < x_k, y_k \leq \max\{\|\v x\|, \|\v y\|\}$ for all $k$. Therefore, there exists a constant $\xi \in (0, \max\{\|\v x\|, \|\v y\|\})$ such 
	that $|\log_{\mathrm{e}} x_k - \log_{\mathrm{e}} y_k| = \xi^{-1}|x_k - y_k| \geq \max\{\|\v x\|, \|\v y\|\}^{-1} |x_k - y_k|$. Since $|\log_{\mathrm{e}} (x_k / y_k)| = |\log_{\mathrm{e}} x_k - \log_{\mathrm{e}} y_k|$,  we have 
	\begin{equation*}
		\begin{aligned}
			\sum_{k \in \mathcal{K}} |x_k - y_k| &\leq \max\{\|\v x\|, \|\v y\|\}K \max_{k \in \mathcal{K}} \left| \log_{\mathrm{e}} \left({x_k}/{y_k} \right) \right| \\ 
			&= \max\{\|\v x\|, \|\v y\|\}K \mu(\v x, \v y). 
		\end{aligned}
	\end{equation*}
	This, together with $\| \v x - \v y \| \leq \| \v x - \v y \|_1$, gives the desired result. 
\end{proof}

\section{Proof of Proposition \ref{prop:error_control}}\label{apd:error_control}
We first show that $\|\widetilde{\vv \beta} - \vv \beta^\star(\v x)\|$ and $\|\widetilde{\v p} - \v p^\star(\widetilde{\vv \beta}, \v x)\|$ are bounded by $\mu(\widetilde{\vv \beta}, I_{\v x}(\widetilde{\vv \beta})) + \mu(\widetilde{\v p}, J_{\widetilde{\vv \beta}, \v x}(\widetilde{\v p}))$ using Lemmas \ref{lem:fixedpoint}--\ref{lem:norm_control} in Appendix \ref{apd:lemmas}. 
Then, we show the desired conclusions using the Lipschitzness of $\widetilde{d}(\cdot, \cdot)$ and $\v g(\cdot, \cdot, \v x)$. 

Let $B_3(\v x) = \max\{\|\widetilde{\vv \beta}\|, \|\vv \beta^\star(\v x)\|, \|\widetilde{\v p}\|, \|\v p^\star(\widetilde{\vv \beta}, \v x)\|\}$. 
Applying Lemma~\ref{lem:norm_control} gives
\begin{equation}
	\begin{aligned}
		\|\widetilde{\vv \beta} - \vv \beta^\star(\v x)\| &\leq B_3(\v x)K \mu(\widetilde{\vv \beta}, \vv \beta^\star(\v x))
	\end{aligned}
	\label{equ:98}
\end{equation}
and 
\begin{equation}
	\begin{aligned}
		\|\widetilde{\v p} - \v p^\star(\widetilde{\vv \beta}, \v x)\| &\leq B_3(\v x)K \mu(\widetilde{\v p}, \v p^\star(\widetilde{\vv \beta}, \v x)). 
	\end{aligned}
	\label{equ:98a}
\end{equation}
Additionally, using the FP property $\vv \beta^\star(\v x) = I_{\v x}(\vv \beta^\star(\v x)) $ and the triangle inequality for Thompson's metric $\mu(\cdot, \cdot)$, we have 
\begin{equation*}
	\begin{aligned}
		\mu(\widetilde{\vv \beta}, \vv \beta^\star(\v x)) &\leq \mu(\widetilde{\vv \beta}, I_{\v x}(\widetilde{\vv \beta})) + \mu(I_{\v x}(\widetilde{\vv \beta}), I_{\v x}(\vv \beta^\star(\v x))) \\
		&\overset{\text{(a)}}{\leq} \mu(\widetilde{\vv \beta}, I_{\v x}(\widetilde{\vv \beta})) + \kappa_2(\v x) \mu(\widetilde{\vv \beta}, \vv \beta^\star(\v x)), 
	\end{aligned}
\end{equation*}
where (a) follows from \eqref{equ:lc1} in Lemma \ref{lem:properties} (b) in Appendix \ref{apd:lemmas}. 
This further implies that
\begin{equation}
	\mu(\widetilde{\vv \beta}, \vv \beta^\star(\v x)) \leq \frac{1}{1-\kappa_2(\v x)} \mu(\widetilde{\vv \beta}, I_{\v x}(\widetilde{\vv \beta})). 
	\label{equ:99}
\end{equation}
Combining \eqref{equ:98} and \eqref{equ:99} gives
\begin{equation}
	\begin{aligned}
		\| \widetilde{\vv \beta} - \vv \beta^\star(\v x) \| 
		&\leq \frac{B_3(\v x)K}{1-\kappa_2(\v x)} \mu(\widetilde{\vv \beta}, I_{\v x}(\widetilde{\vv \beta})).
	\end{aligned}
	\label{equ:res11}
\end{equation}
Similarly, using \eqref{equ:98a}, \eqref{equ:lc_primal_ass} in  Assumption~\ref{ass:lip}~(b), and following the same steps as in \eqref{equ:99} and \eqref{equ:res11} gives 
\begin{equation}
	\|\widetilde{\v p}- \v p^\star(\widetilde{\vv \beta}, \v x)\| \leq \frac{B_3(\v x) K}{1-\kappa_1(\v x)} \mu(\widetilde{\v p}, J_{\widetilde{\vv \beta}, \v x}(\widetilde{\v p})). 
	\label{equ:101}
\end{equation}
Hence, we have
\begin{equation}
	\begin{aligned}
		&\quad\ \| \widetilde{\v p} - \v p^\star(\v x) \| \\
		&\leq \| \widetilde{\v p} - \v p^\star(\widetilde{\vv \beta}, \v x) \| + \|\v p^\star(\widetilde{\vv \beta}, \v x) - \v p^\star(\vv \beta^\star(\v x), \v x) \| \\
		&\overset{\text{(a)}}{\leq} \| \widetilde{\v p} - \v p^\star(\widetilde{\vv \beta}, \v x) \| + L_1(\v x)\|\widetilde{\vv \beta} - \vv \beta^\star(\v x)\| \\
		&\overset{\text{(b)}}{\leq} \frac{B_3(\v x)K}{1-\kappa_1(\v x)} \mu(\widetilde{\v p}, J_{\widetilde{\vv \beta}, \v x}(\widetilde{\v p})) + \frac{L_1(\v x)B_3(\v x)K}{1-\kappa_2(\v x)} \mu(\widetilde{\vv \beta}, I_{\v x}(\widetilde{\vv \beta})), 
	\end{aligned}
	\label{equ:102}
\end{equation}
where (a) is due to \eqref{equ:asss1}, 
and (b) uses \eqref{equ:res11} and \eqref{equ:101}. 

Since $\widetilde{d}(\vv \beta^\star(\v x), \v x)$ in \eqref{equ:func} is a linear function with respect to $\vv \beta^\star(\v x)$ and $\v x$, we have 
	\begin{equation*}
		\|\widetilde d(\widetilde{\vv \beta}, \v x) - \widetilde d(\vv \beta^\star(\v x), \v x)\| \leq L_3 \| \widetilde{\vv \beta} - \vv \beta^\star(\v x) \|, 
		\label{equ:lip_d}
	\end{equation*}
	where $L_3 = \max_{k, m}\{\sigma_k^2, \bar{P}_m\}$.
Using this property together with the bound from \eqref{equ:res11} and noting that $d(\v x) = \widetilde d(\vv \beta^\star(\v x), \v x)$ from \eqref{equ:func}, we obtain 
\begin{equation}
	\begin{aligned}
		\|\widetilde{d}(\widetilde{\vv \beta}, \v x) - d(\v x)\|
		&\leq \frac{L_3 B_3(\v x)K}{1-\kappa_2(\v x)} \mu(\widetilde{\vv \beta}, I(\widetilde{\vv \beta})) \\
		&\leq \frac{L_3 B_3(\v x)K}{1-\kappa_2(\v x)} \mbox{res}_1, 
	\end{aligned}\label{equ:funcv}
\end{equation}
where the second inequality is due to \eqref{equ:res1}. 
In addition, using \eqref{equ:asss2} in Assumption~\ref{ass:lip}~(c) yields
\begin{equation}
	\begin{aligned}
		&\|\v g(\widetilde{\vv \beta}, \widetilde{\v p}, \v x) - \nabla d(\v x)\| \\
		\overset{\text{(a)}}{=}{} & \|\v g(\widetilde{\vv \beta}, \widetilde{\v p}, \v x) - \v g(\vv \beta^\star(\v x), \v p^\star(\v x), \v x)\| \\
		\leq{}& L_2(\v x) \left(\|\widetilde{\vv \beta} - \vv \beta^\star(\v x)\| + \|\widetilde{\v p} - \v p^\star(\v x)\| \right)\\
		\overset{\text{(b)}}{\leq}{}& \frac{L_2(\v x) B_3(\v x) K}{1-\kappa_1(\v x)}\mu(\widetilde{\v p}, J_{\widetilde{\vv \beta}, \v x}(\widetilde{\v p})) \\
		&+ \frac{L_2(\v x) (L_1(\v x) + 1) B_3(\v x) K}{1-\kappa_2(\v x)} \mu(\widetilde{\vv \beta}, I_{\v x}(\widetilde{\vv \beta})) \\
		\overset{\text{(c)}}{\leq}{}& K B_3(\v x)\max\left\{\frac{L_2(\v x) }{1-\kappa_1(\v x)}, \frac{L_2(\v x) (L_1(\v x) + 1) }{1-\kappa_2(\v x)}\right\} \operatorname{res}, 
	\end{aligned}\label{equ:gradd}
\end{equation}
where (a) is due to \eqref{equ:grad2}, (b) follows from \eqref{equ:res11} and \eqref{equ:102}, and (c) uses \eqref{equ:res1} and \eqref{equ:res2}. 
From \eqref{equ:funcv} and \eqref{equ:gradd}, and noting that $\operatorname{res} = \operatorname{res}_1 + \operatorname{res}_2$, we see that \eqref{equ:control_function} and \eqref{equ:control_gradient} hold with 
\begin{equation*}
\begin{aligned}
	C(\v{x}) ={}& KB_3(\v{x}) \max \Bigg\{
	\frac{L_3}{1 - \kappa_2(\v{x})}, 
	\frac{L_2(\v{x})}{1 - \kappa_1(\v{x})}, \\
	&\frac{L_2(\v{x}) (L_1(\v{x}) + 1)}{1 - \kappa_2(\v{x})}
	\Bigg\}.
\end{aligned}
\end{equation*}

\section{Lipschitzness of $\nabla d(\v x)$}
\label{apd:convergece_APIG-FP}
Since $\|\v x^i\| \leq B_1$, i.e., $\{\v x^i\} \subseteq \mathcal{S}_{B_1}^M$, we show the Lipschitzness of $\nabla d(\v x)$ on $\mathcal S_{B_1}^M$ in this section. 
Let $\m D(\v x) = \frac{\partial I_{\v x}}{\partial \vv \beta}(\vv \beta^\star(\v x))$. 
We first prove that $\m I - \m D(\v x)$ is invertable on $\mathcal S_{B_1}^M$, and then we show that this invertibility implies the continuity of $\vv \beta^\star(\cdot)$. 
Finally, we demonstrate the desired conclusion by proving that $\left\|\frac{\partial^2 \vv \beta^\star}{\partial \v x^2}(\v x)\right\| $ is bounded on $\mathcal S_{B_1}^M$. 

First, using \eqref{equ:lc1} with $B(\v x) = 2 \|\vv \beta^\star(\v x)\|$, we have
\begin{equation}
	\begin{aligned}
		&\mu(I_{\v x}((1+t)\vv \beta^\star(\v x)), I_{\v x}(\vv \beta^\star(\v x))) \\
		\leq{} &\kappa_2(\v x) \mu((1+t)\vv \beta^\star(\v x), \vv \beta^\star(\v x)) = \kappa_2(\v x)\log_{\mathrm{e}}(1+t). 
		\label{equ:117}
	\end{aligned}	
\end{equation}
Then, by the definition of Thompson's metric $\mu(\cdot, \cdot)$ in \eqref{equ:metric}, we obtain 
\begin{equation*}
	I_{\v x}((1+t)\vv \beta^\star(\v x)) \leq (1+t)^{\kappa_2(\v x)} I_{\v x}(\vv \beta^\star(\v x)). 
	\label{equ:contraction}
\end{equation*}
This, together with the Taylor expansion of $I_{\v x}(\cdot)$ around $\vv \beta^\star(\v x)$ yields
\begin{equation*}
	\begin{aligned}
		t \m D(\v x) \vv \beta^\star(\v x)
		&\leq ((1+t)^{\kappa_2(\v x)} - 1) I_{\v x}(\vv \beta^\star(\v x)) + o(t) \\
		&\overset{\text{(a)}}{=} ((1+t)^{\kappa_2(\v x)} - 1) \vv \beta^\star(\v x) + o(t), 
	\end{aligned}
\end{equation*}
where (a) uses the fact that $\vv \beta^\star(\v x) = I_{\v x}(\vv \beta^\star(\v x))$.
Dividing both sides of the above equation by $t$ yields
\begin{equation*}
	\begin{aligned}
		\m D(\v x) \vv \beta^\star(\v x) 
		&\leq \frac{(1+t)^{\kappa_2(\v x)} - 1}{t} \vv \beta^\star(\v x) + o(1). 
		%		&\rightarrow \kappa_2(\v x) \vv \beta^\star(\v x) \quad \text{~as~} t\rightarrow 0_+. 
	\end{aligned}
\end{equation*}
Taking the limit as $t \rightarrow 0_+$, we have $\m D(\v x) \vv \beta^\star(\v x) \leq \kappa_2(\v x) \vv \beta^\star(\v x)$. 
In the proof of \cite[Theorem 2]{fan2023QoSbasedBeamformingCompression}), it is shown that $(I_{\v x}(\cdot))_k$ is nondecreasing with respect to each component for any $k\in\cK$ and $\v x \geq 0$, which implies that $\m D(\v x)$ is nonnegative. 
Then, applying \cite[Corollary 8.1.29]{horn1990MatrixAnalysis} gives 
\begin{equation}
	\rho(\m D(\v x)) \leq \kappa_2(\v x) < 1. 
	\label{equ:bound_rho}
\end{equation}
This implies that $\m I - \m D(\v x)$ is invertable on $\mathcal S_{B_1}^M$. 

Next, by Lemma \ref{lem:properties} (a), $I_{\v x}(\cdot)$ is a well-defined rational function on $\mathbb R_+^{M+K}$. 
This ensures the existence and the continuity of the derivatives of $I_{\v x}(\cdot)$ of any order on $\mathbb R_+^{M+K}$. 
Hence, applying the implicit function theorem \cite[Section 8.5]{zorich2004MathematicalAnalysis} to the FP equation \eqref{equ:dual_fp}, we obtain the existence of $\frac{\partial \vv \beta^\star}{\partial \v x}(\v x)$, $\frac{\partial^2 \vv \beta^\star}{\partial \v x^2}(\v x)$, \d. 
The expressions are given by 
$$
\frac{\partial \vv \beta^\star}{\partial \v x}(\v x) = (\m I - \m D(\v x))^{-1}\frac{\partial I_{\v x}}{\partial \v x}(\vv \beta^\star(\v x)), 
$$ 
which shows the continuity of $\vv \beta^\star(\cdot)$ on $\mathcal S_{B_1}^M$; and 
\begin{equation}\label{equ:order2}
	\begin{aligned}
		&\frac{\partial^2 \vv \beta^\star}{\partial \v x^2}(\v x)= (\m I - \m D(\v x))^{-1} \Bigg( \frac{\partial^2 I_{\v x}}{\partial \v x^2} (\vv \beta^\star(\v x)) \\ 
		&+ 2 \frac{\partial^2 I_{\v x}}{\partial \v x \partial \vv \beta} (\vv \beta^\star(\v x)) \frac{\partial \vv \beta^\star}{\partial \v x}(\v x) + \frac{\partial^2 I_{\v x}}{\partial \vv \beta^2} (\vv \beta^\star(\v x)) \Big(\frac{\partial \vv \beta^\star}{\partial \v x}(\v x)\Big)^2 \Bigg). 
	\end{aligned}
\end{equation}

Finally, using the continuity of $\vv \beta^\star(\cdot)$ on $\mathcal S_{B_1}^M$, we have $\max_{\v x \in \mathcal S_{B_1}^M}\vv \beta^\star(\v x) < \infty$.  Hence, applying \eqref{equ:lc2} with $B = \max\{B_1, 2\max_{\v x \in \mathcal S_{B_1}^M}\vv \beta^\star(\v x)\}$ gives
\begin{equation}
	\begin{aligned}
		&\mu(I_{\v x}((1+t)\vv \beta^\star(\v x)), I_{\v x}(\vv \beta^\star(\v x))) \\
		\leq{} &\kappa_4(B) \mu((1+t)\vv \beta^\star(\v x), \vv \beta^\star(\v x)) = \kappa_4(B)\log_{\mathrm{e}}(1+t). 
	\end{aligned}	
	\label{equ:121}
\end{equation}
The only difference between \eqref{equ:117} and \eqref{equ:121} lies in the factor $\kappa_2(\v x)$ (dependent on $\v x$) and $\kappa_4(B)$ (independent of $\v x$). 
Following the same steps as in \eqref{equ:117} and \eqref{equ:bound_rho}, we obtain 
\begin{equation*}
	\rho(\m D(\v x)) \leq \kappa_4(B) < 1,~\forall\,\v x \in \mathcal S_{B_1}^M. 
	\label{equ:bound_rho2}
\end{equation*}
Taking the norm of both sides of \eqref{equ:order2} gives 
\begin{equation*}
	\begin{aligned}
		\left\|\frac{\partial^2 \vv \beta^\star}{\partial \v x^2}(\v x)\right\| \leq \frac{B_4}{1 - \kappa_4(B)} < +\infty,~\forall\,\v x \in \mathcal S_{B_1}^M,
	\end{aligned}
\end{equation*}
where $B_4$ is the upper bound obtained from maximizing the second term in the parenthesis, $\frac{\partial^2 I_{\v x}}{\partial \v x^2} (\vv \beta^\star(\v x)) + 2 \frac{\partial^2 I_{\v x}}{\partial \v x \partial \vv \beta} (\vv \beta^\star(\v x)) \frac{\partial \vv \beta^\star}{\partial \v x}(\v x) + \frac{\partial^2 I_{\v x}}{\partial \vv \beta^2} (\vv \beta^\star(\v x)) \Big(\frac{\partial \vv \beta^\star}{\partial \v x}(\v x)\Big)^2$, in \eqref{equ:order2} on the compact set $\mathcal S_{B_1}^M$. 
Using this and the relationship between $d(\v x)$ and $\vv \beta^\star(\v x)$ given in \eqref{equ:func}, we have
$$
\|\nabla^2 d (\v x)\| \leq \left(\max_{k\in\cK} \sigma_k^2\right) \left\|\frac{\partial^2 \vv \beta^\star}{\partial \v x^2}(\v x)\right\| < +\infty, 
$$
which implies that $\nabla d(\v x)$ is Lipschitz continuous on $\mathcal S_{B_1}^M$. 
Therefore, Assumption \ref{ass:1} is satisfied with $f(\cdot) = -d(\cdot)$. 

\ifCLASSOPTIONcaptionsoff
  \newpage
\fi

\end{document}